\documentclass[12pt]{article}

\usepackage{microtype}
\usepackage{graphicx}
\usepackage{subfigure}
\usepackage{booktabs} 
\usepackage{hyperref}

\usepackage{amsmath}
\usepackage{amssymb}
\usepackage{amsthm}
\usepackage{bm}
\usepackage{graphicx}
\usepackage{natbib}
\usepackage{fullpage}
\usepackage[framed,numbered,autolinebreaks,useliterate]{mcode}

\newtheorem{theorem}{Theorem}

\DeclareMathOperator*{\argmin}{arg\,min}

\begin{document}

\title{Posterior Integration on a Riemannian Manifold}

\author{Chris. J. Oates$^{1,3}$, Alessandro Barp$^{2,3}$ and Mark Girolami$^{2,3}$ \\
$^1$School of Mathematics, Statistics and Physics, Newcastle University, UK \\
$^2$Department of Mathematics, Imperial College London, UK \\
$^3$Alan Turing Institute, UK
}

\maketitle

\begin{abstract}
The geodesic Markov chain Monte Carlo method and its variants enable computation of integrals with respect to a posterior supported on a manifold.
However, for regular integrals, the convergence rate of the ergodic average will be sub-optimal.
To fill this gap, this paper extends the efficient posterior integration method of \cite{Oates2017} to the case of a Riemannian manifold.
In contrast to the original Euclidean case, no non-trivial boundary conditions are needed for a closed manifold.
The method is assessed through simulation and deployed to compute posterior integrals for an Australian Mesozoic paleomagnetic pole model, whose parameters are constrained to lie on the manifold $M = \mathbb{S}^2 \times \mathbb{R}_+$.
\end{abstract}

\section{Introduction}

This work considers numerical approximation of an integral
\begin{eqnarray}
\int_M f \; \mathrm{d}\mathcal{P} \label{eq: integral}
\end{eqnarray}
where $M$ is a $m$-dimensional Riemannian manifold, $\mathcal{P}$ is a distribution suitably defined on $M$ and $f : M \rightarrow \mathbb{R}$ is a $\mathcal{P}$-measurable integrand.
This fundamental problem is well-studied in applied mathematics and existing methods include Gaussian cubatures \citep{Atkinson1982,Filbir2010}, cubatures based on uniformly-weighted quasi Monte Carlo points \citep{Kuo2005,Graef2013} and cubatures based on optimally-weighted Monte Carlo points \citep{Brandolini2010,Ehler2017}.
These numerical integration methods assume that a closed form for $\mathcal{P}$ is provided.
However, this is not the case for many important integrals that occur in the applied statistical context.
In particular, in the Bayesian framework the distribution $\mathcal{P}$ can represent {\it posterior} belief: i.e.
$$
\frac{\mathrm{d}\mathcal{P}}{\mathrm{d}\mathcal{P}_0}(\bm{x}) = \frac{\mathcal{L}(\bm{x})}{Z}
$$
where $\mathcal{P}_0$ is an expert-elicited {\it prior} distribution on $M$ and a function $\mathcal{L}$, known as a {\it likelihood}, determines how the expert's belief should be updated on the basis of data obtained in an experiment; see \cite{Bernardo2001} for the statistical background.
The left hand side of this equation is to be interpreted as the Radon-Nikodym derivative of the posterior with respect to the prior \citep{Stuart2010}.
Outside of conjugate exponential families, posterior distributions are not easily characterised, as the normalisation constant (or {\it marginal likelihood})
\begin{eqnarray*}
Z & = & \int_M \mathcal{L}(\bm{x}) \; \mathrm{d}\mathcal{P}_0(\bm{x}) \label{eq: Z}
\end{eqnarray*}
is itself an intractable integral.
Several methods for approximation of $Z$ have been developed, but this problem is considered difficult -- even in the case of the Euclidean manifold \citep[see the survey in][]{Friel2012}.

Integrals on manifolds arise in many important applications of Bayesian statistics, most notably directional statistics \citep{Mardia2000} and modelling of functional data on the sphere $\mathbb{S}^2$ \citep{Porcu2016}.
The canonical scenario is that $\mathcal{P}_0$ and $\mathcal{P}$ admit densities $\pi_0$ and $\pi_{\mathcal{P}}$ with respect to the natural volume element $\mathrm{d}V$ on the manifold, i.e. $\mathrm{d}\mathcal{P}_0 = \pi_0 \mathrm{d} V$ and $\mathrm{d}\mathcal{P} = \pi_{\mathcal{P}} \mathrm{d} V$, and that a function $\pi$ proportional to $\pi_{\mathcal{P}}$ can be provided.
Specifically, in the context of Bayesian statistics on a Riemannian manifold, we are provided with $\pi(\bm{x}) = \mathcal{L}(\bm{x}) \pi_0(\bm{x})$.
To facilitate approximation of Eqn. \ref{eq: integral} in this context, Markov chain Monte Carlo (MCMC) methods have been developed to sample from distributions defined on a manifold \citep{Byrne2013,Lan2014,Holbrook2016}.
Their output is a realisation of an ergodic Markov process $(\bm{x}_i)_{i=1}^n$ that leaves $\mathcal{P}$ invariant, so that the integral may be approximated by an ergodic average $\frac{1}{n} \sum_{i=1}^n f(\bm{x}_i)$.
Convergence of MCMC estimators is well-understood \citep{Meyn2012}.

A drawback of MCMC, of practical significance for regular problems of modest dimension, is that the convergence is gated at $O_P(n^{-\frac{1}{2}})$.
This rate is inferior to the rates obtained by the aforementioned methods that apply when $\mathcal{P}$ is explicit; a consequence of the fact that the ergodic average does not exploit smoothness of the integrand.
In recent years, several alternatives to MCMC have been developed to address this convergence bottleneck.
These include transport maps \citep{Marzouk2016}, Riemann sums \citep{Philippe2001}, quasi Monte Carlo ratio estimators \citep{Schwab2012} and estimators based on Stein's method \citep{Liu2016,Oates2017,Oates2018}.
However, these methods have so far focused on the case of the Euclidean manifold $M = \mathbb{R}^d$.

In this paper we generalise one of these methods -- the method proposed in \cite{Oates2017} -- for computation of posterior integrals on a Riemannian manifold.
Inspired by classical integration methods, our approach is based on approximation of the integrand: 
In the first step, the un-normalised density $\pi$ is exploited to construct a class $H$ of functions, defined on the manifold, that can be exactly integrated with respect to $\mathcal{P}$. 
Next, the integrand $f$ is approximated with a suitably chosen element $\hat{f}$ from $H$. 
Finally, an approximation to Eqn. \ref{eq: integral} is provided by $\int_M \hat{f} \mathrm{d} \mathcal{P}$.
The main technical contribution occurs in the first step, where we must elucidate a class $H$ of functions that can be integrated without access to $Z$, the normalisation constant.

The main properties of the proposed method, which hold also for the Euclidean case \cite{Oates2017,Oates2018}, are as follows:
\begin{itemize}
\item The convergence rate is empirically verified at $o(n^{-\frac{1}{2}})$, under regularity assumptions on the integrand.
On the other hand, the computational cost associated with the estimator is up to $O(n^3)$.
\item The points $\{\bm{x}_i\}_{i=1}^n$ at which the integrand is evaluated {\it do not} need to form an approximation to $\mathcal{P}$.
\item A computable upper bound on (relative) integration error -- a {\it kernel Stein discrepancy} \citep{Chwialkowski2016,Liu2016b,Gorham2017} -- is obtained as a by-product of approximating the integral.
\end{itemize}
Moreover, in this paper we demonstrate that, compared to the case of the Euclidean manifold, non-trivial boundary conditions are not required for the method to be applied on a closed manifold.

The paper proceeds as follows:
In Sec. \ref{sec: background} we provide a brief mathematical background.
In Sec. \ref{sec: method} we present the proposed method.
The method is empirically assessed in Sec. \ref{sec: results}.
Further discussion of the approach is provided in Sec. \ref{sec: discussion}.

\section{Mathematical Background} \label{sec: background}

The aim of this section is to present an informal and accessible introduction to some of the mathematical tools that are needed for our development.
For a formal treatment, several references to textbooks are provided.

\paragraph{Embedded Riemannian Manifolds}

Our presentation focuses on embedded manifolds, as any abstract manifold may be embedded in $\mathbb R^d$ for some $d \in \mathbb{N}$ by the Whitney embedding theorem \citep{Skopenkov2008}. 
However, the method itself will not require an embedding to be explicit.
Recall the space $\mathbb R^d$ may be equipped with global coordinates $\bm{x} = (x_1, \dots ,x_d)$ and the natural inner product. 
An $m$-dimensional manifold $M$ embedded in $\mathbb R^d$ is a subset of $\mathbb R^d$ such that any point $\bm{x} \in M$ has a neighbourhood $O_{\bm{x}}\subset M$ which can be parameterised by local coordinates $\bm{q}=(q_1, \dots ,q_m) \in Q_{\bm{x}} \subseteq \mathbb{R}^m$; i.e. there exists a smooth map $\nu_{\bm{x}} : Q_{\bm{x}} \rightarrow O_{\bm{x}}$ with smooth inverse $\nu_{\bm{x}}^{-1}$. 

A tangent vector to $M$ at a point $\bm{x}$ is defined as the tangent vector at $\bm{x}$ to a curve on $M$. 
Since a curve on $M$ may be locally parametrised as $\bm{\gamma}(t)=\nu_{\bm{x}}(\bm{q}(t))$, its tangent vector is 
$$
\underline{\bm{\gamma}}' \; = \; \sum_{i=1}^m \frac{\mathrm{d} q_i(t)}{\mathrm{d} t} \underbrace{\frac{ \partial \nu_{\bm{x}} (\bm{q})}{\partial q_i}}_{\partial_{q_i}} .
$$ 
Thus the $m$ vectors $\partial_{q_i} \in \mathbb{R}^d$ form a basis for the tangent space, denoted $T_{\bm{x}} M$. 
The tangent space is equipped with the inner product $\langle \underline{\bm{a}} , \underline{\bm{b}} \rangle_{\mathrm{G}(\bm{x})} = \underline{\bm{a}}^\top \mathrm{G}(\bm{x}) \underline{\bm{b}}$, where $\mathrm{G}_{ij}(\bm{x}) = \langle \partial_{q_i} , \partial_{q_j} \rangle$, that arises as the restriction of the usual inner product on $\mathbb{R}^d$. 
The pair $(M,\mathrm{G})$ is called a Riemannian manifold.

{\it Example:} The sphere $\mathbb{S}^2$ is a Riemannian manifold.
The coordinate patch $\nu_{\bm{x}}(\bm{q})=(\cos q_1 \sin q_2 , \sin q_1 \sin q_2 , \cos q_2)$, with local coordinates $q_1 \in (0,2\pi)$, $q_2 \in (0,\pi)$, holds for almost all\footnote{It does not cover the half great circle that passes through both poles and the point $(1,0,0)$.} $\bm{x} \in \mathbb{S}^2$. 
The tangent space is spanned by $\partial_{q_1}  = (-\sin q_1 \sin q_2, \cos q_1 \sin q_2, 0)$ and $\partial_{q_2} = (\cos q_1 \cos q_2, \sin q_1 \cos q_2, -\sin q_2)$. 
Taking the Euclidean inner product of these vectors produces $\mathrm{G}_{1,1}= \sin^2 q_2$, $\mathrm{G}_{2,2} = 1$, $\mathrm{G}_{1,2} = \mathrm{G}_{2,1} = 0$.

\paragraph{Geometric Measure Theory}

Any oriented Riemannian manifold has a natural measure $\mathrm{d}V$ over its Borel algebra, called the Riemannian volume form, which represent an infinitesimal volume element. 
In a coordinate patch $Q_{\bm{x}} \subset \mathbb{R}^m$, this measure can be expressed in terms of the Lebesgue measure: $\mathrm{d}V =\sqrt{\text{det}(\mathrm{G}(\bm{x}))} \lambda^m(\mathrm{d}\bm{q})$. 
In particular when $M$ is the Euclidean space, this is just the Lebesgue measure, and when $M$ is an embedded manifold in $\mathbb R^d$, $\mathrm{d}V$ is the surface area (or Hausdorff) measure $\mathcal{H}(\mathrm{d}\bm{x})$ \citep{Federer1969}. 

{\it Example:} For the sphere $\mathbb{S}^2$, $\mathrm{d}V = \sin q_2 \mathrm{d} q_1 \mathrm{d} q_2$, where $\sin q_2$ is the area of the parallelogram spanned by $\partial_{q_1}, \partial_{q_2}$.

A technical point is that we restrict attention to Riemannian manifolds that are \emph{oriented}.
This is equivalent to assuming that the volume form $\mathrm{d}V$ is coordinate independent. 
It will also be required that $M$ is either closed or is a \emph{manifold with boundary} $\partial M$ \citep[see p25 of][]{Lee2013}.

\paragraph{Calculus on a Riemannian Manifold}

To present a natural, coordinate-independent construction of differential operators on manifolds would require either exterior calculus or the concept of a covariant derivative. 
To limit scope, we present two important differential operators in local coordinates and merely comment that the associated operators are in fact coordinate-independent; full details can be found in \cite{Bachman2006}. 
To this end, denote the gradient of a function $\phi : M \rightarrow \mathbb{R}$, assumed to exist, as
$$
\nabla \phi = \sum_{i,j = 1}^m [\mathrm{G}^{-1}]_{i,j} \frac{\partial \phi}{\partial q_j} \partial_{q_i}
$$
Likewise, define the divergence of a vector field $\underline{\bm{s}} = s_1 \partial_{q_1} + \dots + s_m  \partial_{q_m}$ with $s_i = s_i(\bm{x})$, assumed to exist, as
$$
\nabla \cdot \underline{\bm{s}} = \sum_{i=1}^m \frac{\partial s_i}{\partial q_i} + s_i \frac{\partial}{\partial q_i} \log \sqrt{\text{det}(\mathrm{G})}.
$$
These two differential operators are sufficient for our work; for instance, they can be combined to obtain the Laplacian $\Delta \phi := \nabla \cdot \nabla \phi$.

\paragraph{Divergence Theorem on a Manifold}

The divergence theorem on an oriented Riemannian manifold with boundary $\partial M$ states that: 
\begin{eqnarray*}
\int_M \nabla \cdot \underline{\bm{s}} \; \mathrm{d}V & = & \int_{\partial M} \langle \underline{\bm{s}}, \underline{\bm{n}} \rangle_{\mathrm{G}} \; i_{\underline{\bm{n}}} \mathrm{d}V 
\end{eqnarray*}
where $i_{\underline{\bm{n}}} \mathrm{d}V$ is the volume form on the boundary $\partial M$ and $\underline{\bm{n}}$ is the unit normal vector pointing outward \citep{Bachman2006}.
To define $\underline{\bm{n}}$ one uses the fact that, if $\tilde{M}$ is a Riemannian submanifold of $M$, then for each $\bm{x} \in \tilde{M}$, the metric $\mathrm{G}$ of $M$ splits the tangent space $T_{\bm{x}}M$ into $T_{\bm{x}} \tilde{M}$ and its orthogonal complement $N_{\bm{x}}$; i.e. $T_{\bm{x}} M = T_{\bm{x}} \tilde{M} \oplus N_{\bm{x}}$. 
Elements of $N_{\bm{x}}$ are normal vectors to $\tilde{M}$.
To define $i_{\underline{\bm{n}}} \mathrm{d}V$, note that $\partial M$ is a submanifold of $M$ and the restriction $\left. \mathrm{G} \right|_{\partial M}$ of the metric $\mathrm{G}$ induces a Riemannian mainfold $(\partial M , \left. \mathrm{G} \right|_{\partial M})$.
Then $i_{\underline{\bm{n}}} \mathrm{d}V$ can be seen as the natural volume form on the induced manifold.

Thus if $M$ is a closed manifold, then $\int_M \nabla \cdot \underline{\bm{s}} \; \mathrm{d}V = 0$.
This fact will allow for considerable simplification of the proposed approach, compared to the Euclidean case studied in \cite{Oates2017,Oates2018}, where non-trivial boundary conditions were required.

\paragraph{Reproducing Kernel Hilbert Spaces}

The definition of a (real-valued) reproducing kernel Hilbert space (RKHS) on the manifold $M$ is identical to the usual definition on the Euclidean manifold.
Namely, an RKHS is a Hilbert space $(H,\langle \cdot , \cdot \rangle_H )$ of functions $H \ni h : M \rightarrow \mathbb{R}$ equipped with a {\it kernel}, $k : M \times M \rightarrow \mathbb{R}$, that satisfies: (a) $k(\cdot,\bm{x}) \in H$ for all $\bm{x} \in M$; (b) $k(\bm{x},\bm{y}) = k(\bm{y},\bm{x})$ for all $\bm{x},\bm{y} \in M$; (c) $\langle h , k(\cdot,\bm{x}) \rangle_H = h(\bm{x})$ for all $h \in H$, $\bm{x} \in M$.
For standard manifolds, such as the sphere $\mathbb{S}^2$, several function spaces and their reproducing kernels have been studied \citep[e.g.][]{Porcu2016}.
For more general manifolds, an extrinsic kernel can be induced from restriction under embedding into an ambient space \citep{Lin2017}, or the stochastic partial differential approach \citep{Fasshauer2011,Lindgren2011,Niu2017} can be used to numerically approximate a suitable intrinsic kernel.

Three important facts will be used later:
First, the kernel $k$ characterises the inner product $\langle \cdot , \cdot \rangle_H$ and the set $H$ consists of functions $h$ with finite norm $\|h\|_H = \langle h , h \rangle_H^{1/2}$.
Second, if $H$ and $\tilde{H}$ are two RKHS on $M$ with reproducing kernels $k$ and $\tilde{k}$, then $H + \tilde{H}$ can be defined as the RKHS whose elements can be written as $h + \tilde{h}$, $h \in H$, $\tilde{h} \in \tilde{H}$, with reproducing kernel $k + \tilde{k}$.
Third, if $L$ is a linear operator and $H$ is an RKHS, then $L H$ can be defined as the set $LH = \{L(h) : h \in H\}$ endowed with the reproducing kernel $L \bar{L} k(\bm{x},\bm{y})$, where $\bar{L}$ denotes the adjoint of the operator $L$, which acts on the second argument $\bm{y}$ rather than the first argument $\bm{x}$.
See \cite{Berlinet2011} for several examples of RKHS and additional technical background.

This completes our brief tour of the mathematical prerequisites; the next section describes the proposed posterior integration method.

\section{Posterior Integration on a Manifold} \label{sec: method}

In this section we present the proposed numerical integration method, which proceeds in three steps:
\begin{enumerate}
\item Construct a flexible class $H$ of functions $h : M \rightarrow \mathbb{R}$ such that the integrals $\int_M h \; \mathrm{d}\mathcal{P}$ with respect to $\mathcal{P}$ can be exactly computed.
\item Approximate the integrand $f$ with a suitably chosen element $\hat{f}$ from $H$.
\item Approximate the integral of interest as $\int_M \hat{f} \; \mathrm{d} \mathcal{P}$.
\end{enumerate}
It is clear that step 1 is non-trivial, since availability of the normalisation constant $Z$ in Eqn. \ref{eq: Z} cannot be assumed.
This will be our focus next.

\paragraph{Step \# 1: Constructing an Approximating Class $H$}

To proceed, we generalise the method of \cite{Oates2017} to the case of an oriented Riemannian manifold.

Let $\Phi$ be a RKHS of twice differentiable functions $\phi : M \rightarrow \mathbb{R}$ whose reproducing kernel is denoted $k$.
Then $\underline{\bm{s}} = \pi \nabla \phi$ denotes a gradient field on $M$ and we may consider its divergence $\nabla \cdot \underline{\bm{s}}$ on $M$.
In particular, consider the linear differential operator 
$$
L_{\pi}(\phi) \; := \; \frac{\nabla \cdot (\pi \nabla \phi)}{\pi}.
$$
The proposed method rests on the following result, which is proven in Sec. \ref{ap: proof sec} of the Supplement:
\begin{theorem} \label{thm: main}
For $\phi \in \Phi$ it holds that
\begin{eqnarray*}
\int_M L_\pi(\phi) \mathrm{d} \mathcal{P} & = & \frac{1}{Z} \int_{\partial M} \langle \pi \nabla \phi , \underline{\bm{n}} \rangle_{\mathrm{G}} \; i_{\underline{\bm{n}}} \mathrm{d}V 
\end{eqnarray*}
Thus, if $M$ is a closed manifold, the right hand side vanishes and $L_\pi (\phi)$ can be trivially integrated.
\end{theorem}

Note that the same conclusion holds even when $M$ is not closed, provided that $\langle \pi \nabla \phi , \underline{\bm{n}} \rangle_{\mathrm{G}}$ vanishes everywhere on the boundary $\partial M$; this is similar to the non-trivial assumption made in \cite{Oates2017} for the case of the Euclidean manifold.
Note that $L_\pi$ is not the only differential operator that could be used; others are suggested in Sec. \ref{ap: stein operators} of the Supplement.

The RKHS $H_\pi = L_\pi \Phi$, whose elements are functions of the form $L_\pi(\phi)$ and whose kernel is $k_\pi = L_\pi \bar{L}_\pi k$, is not quite flexible enough for our purposes, since these cannot approximate the constant function $f(\bm{x}) = 1$.
Thus we augment $L_\pi \Phi$ with the RKHS of constant functions, denoted $\{\sigma\}$ and equipped with constant kernel $\sigma^2$, to obtain the function class $H_{\pi,\sigma} = \{\sigma\} + L_\pi \Phi$. 
Of course, the integral of the constant function with respect to $\mathcal{P}$ is trivially computed as $\mathcal{P}$ is a probability distribution.
It follows that $H_{\pi,\sigma}$ is a RKHS with kernel $k_{\pi,\sigma}(\bm{x},\bm{y}) = \sigma^2 + k_\pi$. 
To ensure the elements of $H_{\pi,\sigma}$ integrate to zero under $\mathcal{P}$, from Theorem \ref{thm: main}, we therefore require that $\langle \pi \nabla k(\cdot,\bm{y}) , \underline{\bm{n}} \rangle_{\mathrm{G}}$ vanishes on $\partial M$ for each fixed $\bm{y} \in M$ whenever $M$ is not closed.

Under certain regularity assumptions, and additional technical details to deal with the fact that a slightly different differential operator was used, the set $H_{\pi,\sigma}$ can be shown to be dense in $L_2(\mathcal{P})$ in the case of the Euclidean manifold \citep[c.f. Lemma 4 of][]{Oates2018}.

\paragraph{Step \#2: Approximating the Integrand}

Now that we have a class of functions $H_{\pi,\sigma}$ that can be exactly integrated, we must attempt to approximate $f$ with an element from this set.
Following \cite{Oates2017}, the estimator that we consider is
$$
\hat{f} = \argmin_{h \in H_{\pi,\sigma}} \|h\|_{H_{\pi,\sigma}} \; \text{s.t.} \; h(\bm{x}_i) = f(\bm{x}_i), \; i \in \{ 1,\dots,n \} .
$$
From the representer theorem \citep[see e.g.][]{Scholkopf2001} it follows that $\hat{f}$ has a closed form expression in terms of the kernel $k_{\pi,\sigma}$:
\begin{eqnarray*}
\hat{f}(\cdot) & = & [k_{\pi,\sigma}(\cdot,\bm{x}_1) \dots k_{\pi,\sigma}(\cdot,\bm{x}_n)] \times  { \underbrace{\left[  \begin{array}{ccc} k_{\pi,\sigma}(\bm{x}_1,\bm{x}_1) & \dots & k_{\pi,\sigma}(\bm{x}_1,\bm{x}_n) \\ \vdots & & \vdots \\ k_{\pi,\sigma}(\bm{x}_n,\bm{x}_1) & \dots & k_{\pi,\sigma}(\bm{x}_n,\bm{x}_n) \end{array} \right]}_{\mathbf{K}_{\pi,\sigma}} }^{-1} \underbrace{\left[ \begin{array}{c} f(\bm{x}_1) \\ \vdots \\ f(\bm{x}_n) \end{array} \right]}_{\mathbf{f}}
\end{eqnarray*}
This has the form of a weighted combination of functions in $H_{\pi,\sigma}$:
$$
\hat{f}(\cdot) = \sum_{i=1}^n w_i k_{\pi,\sigma}(\cdot,\bm{x}_i), \qquad w_i = [\mathbf{K}_{\pi,\sigma}^{-1} \mathbf{f}]_i
$$
The form of the estimator $\hat{f}$ is rather standard and can be characterised in several ways, e.g. as a Bayes rule for an $L_2$ regression problem or as a posterior mean under a suitable Gaussian process regression model.
Alternative kernel estimators, such as estimators that enforce non-negativity of the weights $w_i$, could be considered \citep[c.f.][]{Liu2017b,Ehler2017}.

\paragraph{Step \#3: Approximating the Integral}

The approximation $\hat{f}$ can be exactly integrated by construction:
\begin{eqnarray}
\int_M \hat{f} \; \mathrm{d} \mathcal{P} & = & \sum_{i=1}^n w_i \int_M k_{\pi,\sigma}(\cdot , \bm{x}_i) \; \mathrm{d} \mathcal{P} \nonumber \\
& = &  \sum_{i=1}^n w_i \sigma^2 \; = \; \sigma^2 \mathbf{1}^\top \mathbf{K}_{\pi,\sigma}^{-1} \mathbf{f} \label{eq: kernel quadrature}
\end{eqnarray}
The estimate in Eqn. \ref{eq: kernel quadrature} is recognised as a kernel quadrature method and, as such, it carries a Bayesian interpretation \citep{Briol2016}.
Namely, from the Bayesian perspective, Eqn. \ref{eq: kernel quadrature} is the posterior mean for $\int_M f \mathrm{d}\mathcal{P}$ when $f$ is modelled {\it a priori} as a centred Gaussian process with covariance function $k_{\pi,\sigma}$ \citep[see][for background on Gaussian process models]{Rasmussen2006}.
In this light, the parameter $\sigma$ can be considered as a prior standard deviation for the value of the integral $\int_M f \mathrm{d} \mathcal{P}$.
Thus, since we may not know the size of the values taken by $f$ in advance, we consider a weakly informative prior corresponding to the limit $\sigma \rightarrow \infty$.
To this end, we have the following, proven in Sec. \ref{ap: proof sec} of the Supplement:
\begin{theorem} \label{thm: limit}
Let $\mathbf{K}_\pi$ denote the kernel matrix with entries $k_\pi(\bm{x}_i,\bm{x}_j)$.
Then
\begin{eqnarray}
\lim_{\sigma \rightarrow \infty} \int_M \hat{f} \; \mathrm{d}\mathcal{P}
& = & \left( \frac{\mathbf{K}_\pi^{-1} \mathbf{1}}{\mathbf{1}^\top \mathbf{K}_\pi^{-1} \mathbf{1}} \right)^\top \mathbf{f} . \label{eq: estimator}
\end{eqnarray}
\end{theorem}
The estimator in Eqn. \ref{eq: estimator} is the one that is experimentally tested in Section \ref{sec: results}.
An interesting observation is that the weights $w_i$ automatically sum to unity in this approach.
Moreover, the expression $(\mathbf{1}^\top \mathbf{K}_\pi^{-1} \mathbf{1})^{-1/2}$ is exactly the worst case error of the weighted point set $\{(w_i,\bm{x}_i)\}_{i=1}^n$ in the unit ball of $H_\pi$:
\begin{eqnarray}
(\mathbf{1}^\top \mathbf{K}_\pi^{-1} \mathbf{1})^{-1/2} & = &  \sup \left\{ \left| \int_M f \mathrm{d} \mathcal{P} - \lim_{\sigma \rightarrow \infty} \int_M \hat{f} \mathrm{d}\mathcal{P} \right| \; : \; \|f\|_{H_\pi} \leq 1 \right\} \label{eq: wce} 
\end{eqnarray}
This quantity is also known as the {\it kernel Stein discrepancy} associated with the weighted point set $\{(w_i,\bm{x}_i)\}_{i=1}^n$ \citep{Chwialkowski2016,Liu2016b,Gorham2017}.
Thus a measure of (relative) integration error comes for free when the estimator is computed.
From standard duality, this expression is also the posterior standard deviation for the integral.
Of course, in practice the linear system $\mathbf{K}_\pi^{-1} \mathbf{1}$ need only be solved once, at a cost of at most $O(n^3)$.

\paragraph{Related Work}

The original work of \cite{Oates2017} considered an arbitrary vector field $\bm{\phi}$ in place of the gradient field $\nabla \phi$, and thus required only a first order differential operator.
This was possible since the coordinates of the vector field could be dealt with independently in the case of the Euclidean manifold, but this will not be possible in the case of a general manifold.
Interestingly, the second order differential operator considered here is the manifold generalisation of the operator used in the earlier work of \cite{Assaraf1999,Mira2013} and recently rediscovered in the context of Riemannian Stein variational gradient descent in \cite{Liu2017}.
The latter reference is most similar to our work, but focused on construction of a point set as opposed to the question of how to construct an estimator based on a given point set.
Other {\it Stein operators} for the Euclidean manifold were proposed in \cite{Gorham2016}.
The divergence theorem was recently also used in order to generalise the score matching method for parameter estimation on a Riemannian manifold in \cite{Mardia2016}.

\section{Numerical Assessment} \label{sec: results}

In this section we report experiments designed to assess the performance of the proposed numerical method.
In Sec. \ref{subsec: Euclidean} we return to the standard case of the Euclidean manifold $M = \mathbb{R}^d$, then in Sec. \ref{subsec: sphere} we present experiments performed on the sphere $M = \mathbb{S}^2$.
Last, in Sec. \ref{subsec: Application} we applied the proposed method to an Australian Mesozoic paleomagnetic pole model where the manifold was $M = \mathbb{S}^2 \times \mathbb{R}_+$.

\subsection{Euclidean Manifold} \label{subsec: Euclidean}

First, we considered the Euclidean manifold $M = \mathbb{R}^d$.
This was for two reasons; first, to expose the proposed construction in a familiar context, and second, to determine whether the use of a second order differential operator leads to any substantive differences relative to earlier work.

\paragraph{Differential Operator}

For $M = \mathbb{R}^d$, we have a global parametrisation $\bm{q} = \bm{x}$ and the natural volume form is the Lebesgue measure; $\mathrm{d} V = \mathrm{d} \lambda^d$.
For simplicity, suppose that either $\pi$ vanishes on $\partial M$ or $M = \mathbb{R}^d$.
Then our method involves the second order differential operator
\begin{eqnarray*}
L_\pi(\phi) \; = \; \frac{\nabla \cdot (\pi \nabla \phi)}{\pi} \; = \; \frac{\nabla \pi}{\pi} \cdot \nabla \phi + \Delta \phi
\end{eqnarray*}
where $\nabla$ is the familiar gradient.
For the case where $M$ is bounded, let $\bm{n}(\bm{x})$ denote the unit normal to $\partial M$.
Then from the Euclidean version of the divergence theorem:
\begin{eqnarray*}
\int_M L_\pi(\phi) \; \mathrm{d} \mathcal{P} & = & \int_{\partial M} \pi(\bm{x}) \nabla \phi(\bm{x}) \cdot \bm{n}(\bm{x}) \; \mathrm{d}\lambda^d(\bm{x}) \\
& = & \int_{\partial M} 0 \; \mathrm{d}\lambda^d(\bm{x}) \quad = \quad 0.
\end{eqnarray*}
The equality is immediate for the case $M = \mathbb{R}^d$.

This is to be contrasted with the earlier work of \cite{Oates2017}, which considered a general vector field $\bm{\phi} : M \rightarrow \mathbb{R}^d$ and the first order differential operator $L_\pi^1 (\bm{\phi}) = \frac{1}{\pi} \nabla \cdot (\pi \bm{\phi})$.
From there, \cite{Oates2017} proceed as we have already described, with $\bm{\phi} \in \Phi \times \dots \times \Phi$ the tensor product of $d$ copies of the RKHS $\Phi$.
Note that $L_\pi^1$ implicitly relies on the Euclidean structure of the manifold and cannot be general.

\paragraph{Choice of Kernel}

If $\alpha \in \mathbb{N} + \frac{1}{2}$ then the Mat\'{e}rn kernel 
\begin{eqnarray*}
k(\bm{x},\bm{y}) & = & \lambda^2 \exp\left( - \frac{\sqrt{2\alpha} \|\bm{x} - \bm{y}\|}{\ell} \right) \frac{\Gamma(\alpha + \frac{1}{2})}{\Gamma(2 \alpha)} \sum_{i=0}^{\alpha - \frac{1}{2}} \frac{(\alpha - \frac{1}{2} + i)!}{i! (\alpha - \frac{1}{2} - i)!} \left( \frac{\sqrt{8\alpha}\|\bm{x} - \bm{y}\|}{\ell} \right)^{\alpha - \frac{1}{2} - i} 
\end{eqnarray*}
with parameters $\lambda , \ell > 0$ reproduces the Sobolev space $\Phi = H^\alpha(\mathbb{R}^d)$ \citep[see e.g.][]{Fasshauer2007}.
In order for $L_\pi \Phi$ to be well-defined, we require that elements of $\Phi$ are twice (weakly) differentiable; hence we require that $\alpha > 2$.
In contrast, for $L_\pi^1$ to be well defined we have the weaker requirement that $\alpha > 1$.

\paragraph{Experimental Results}

For a transparent and reproducible test bed, let $\mathcal{P}$ be the standard Gaussian in $\mathbb{R}^d$ and suppose that we are told $\pi(\bm{x}) = \exp(- \frac{1}{2} \|\bm{x}\|^2 )$.
That is, we pretend that the normalisation constant $(2 \pi)^{\frac{d}{2}}$ is unknown and proceed as described.
For the kernel we fixed\footnote{All kernel parameters were fixed at sensible default values in this work, but optimisation of these parameters can be expected to offer improvement.} $\lambda = 1$, $\ell = 1$ and considered values $\alpha \in \{\frac{3}{2}, \frac{5}{2} , \frac{7}{2}\}$. 
For the points $\{\bm{x}_i\}_{i=1}^n$ we consider three scenarios:
\begin{enumerate}
\item independent, unbiased draws $\bm{x}_i \sim \mathcal{N}(\mathbf{0},\mathbf{I})$
\item independent, biased draws $\bm{x}_i \sim \mathcal{N}(\mathbf{1},3\mathbf{I})$
\item (for $d=1$) stratified points $\bm{x}_i = $ the $\frac{i}{n+1}$th percentile of $\mathcal{N}(0,1)$
\end{enumerate}
For each point set we computed the worst case error in Eqn. \ref{eq: wce}.
(For scenarios 1 and 2 we report the mean worst case error obtained over 100 independent realisations of the random points.)
Both differential operators $L_\pi$ and $L_\pi^1$ were considered, although the former is incompatible with $\alpha = \frac{3}{2}$.
All results were obtained using MATLAB R2017b, with symbolic differentiation exploited to compute all kernels $k_\pi$.
Further details are provided in Sec. \ref{ap: sym diff} of the Supplement.

\begin{figure}[t!]
\centering
\includegraphics[width = 0.33\textwidth,clip,trim = 6cm 10.5cm 6cm 11cm]{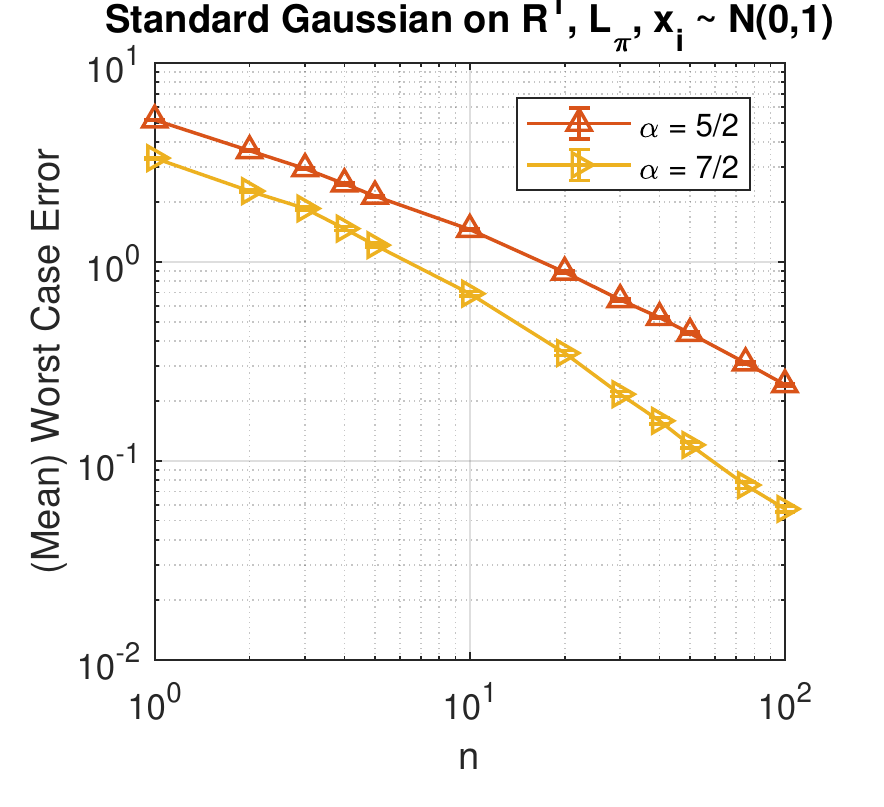}
\includegraphics[width = 0.33\textwidth,clip,trim = 6cm 10.5cm 6cm 11cm]{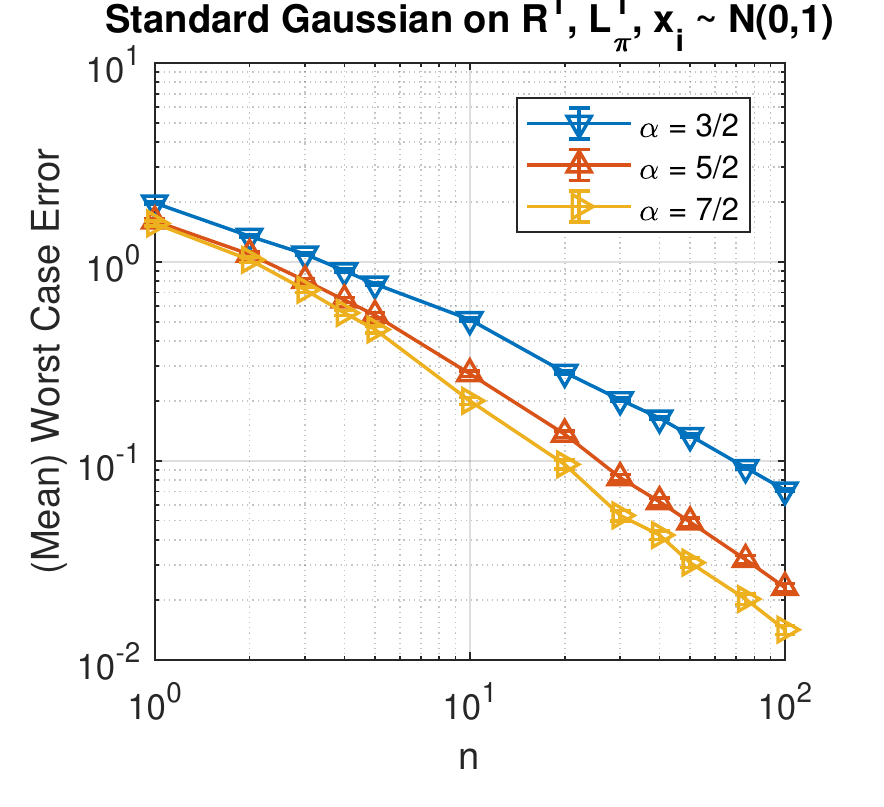}

\includegraphics[width = 0.33\textwidth,clip,trim = 6cm 10.5cm 6cm 11cm]{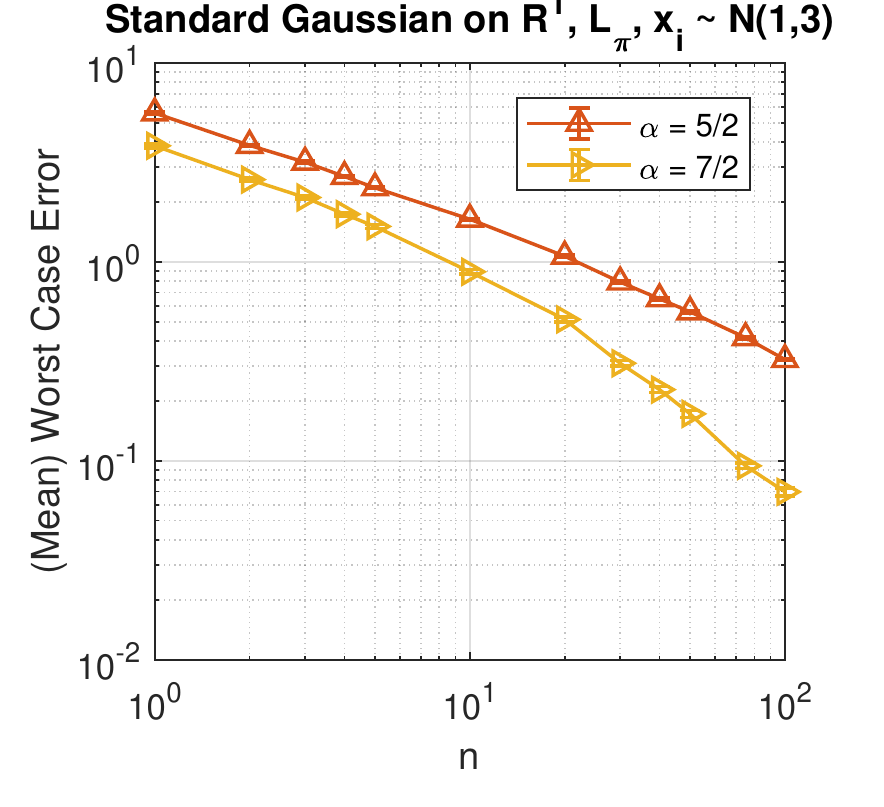}
\includegraphics[width = 0.33\textwidth,clip,trim = 6cm 10.5cm 6cm 11cm]{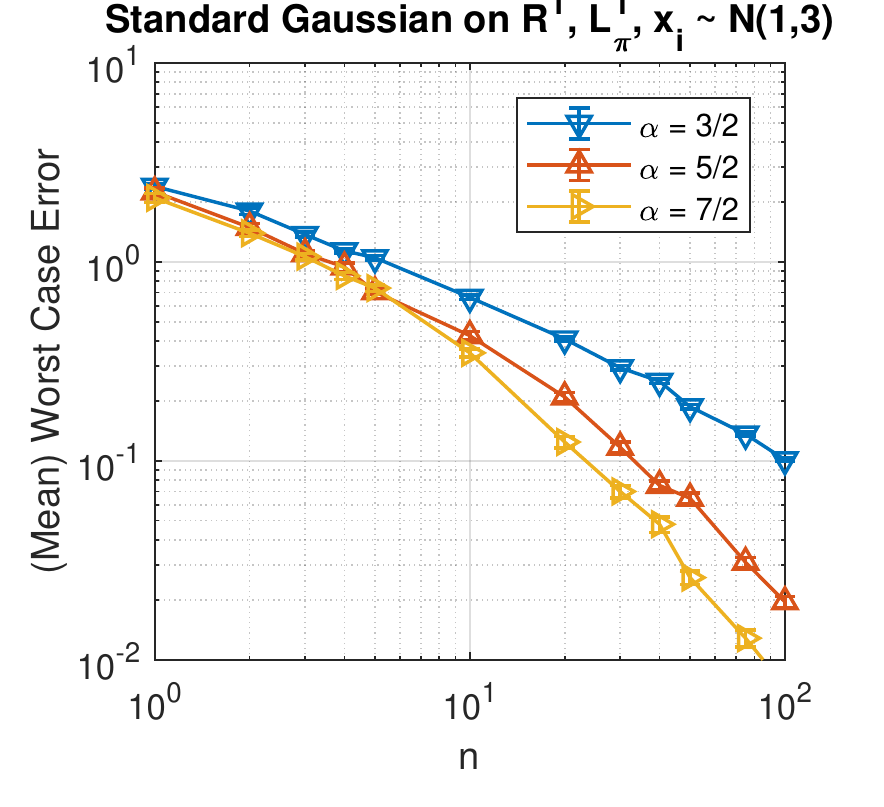}

\includegraphics[width = 0.33\textwidth,clip,trim = 6cm 10.5cm 6cm 11cm]{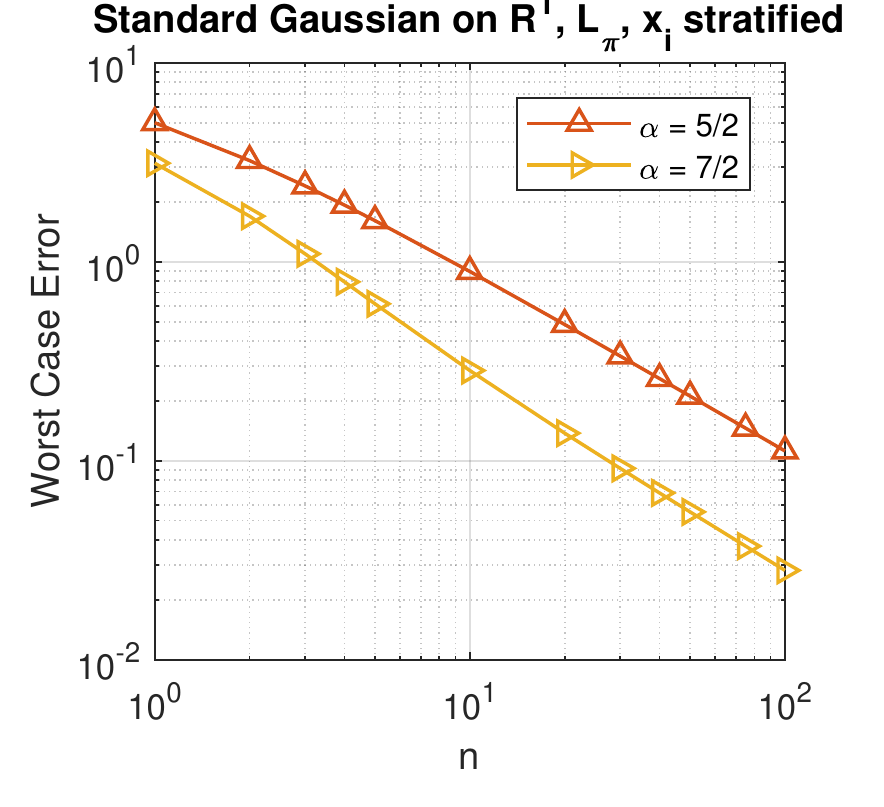}
\includegraphics[width = 0.33\textwidth,clip,trim = 6cm 10.5cm 6cm 11cm]{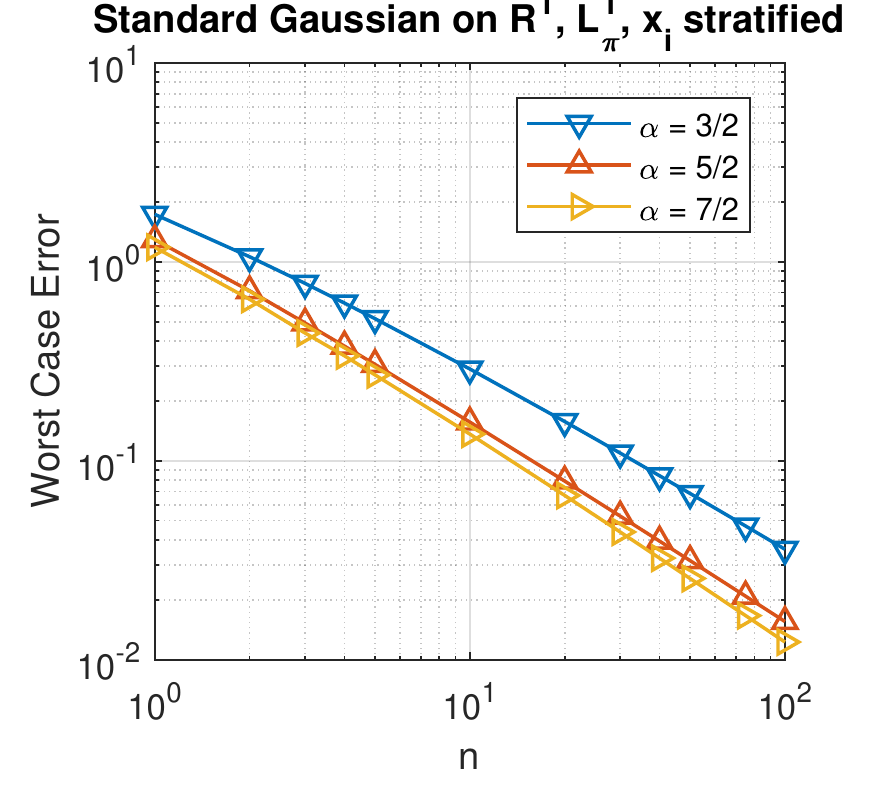}

\caption{Results for the standard Gaussian on the Euclidean manifold in dimension $d = 1$.
The worst case error of the proposed method (Eqn. \ref{eq: wce}) was plotted for various $\alpha$, controlling the smoothness of the kernel, and various $n$, the number of evaluations of the integrand.
[Top row: The points $\bm{x}_i \sim \mathcal{N}(0,1)$ were drawn from the target.
Middle row: The points $\bm{x}_i \sim \mathcal{N}(1,3)$ were drawn from an incorrect distribution.
Bottom row: The points $\bm{x}_i$ were stratified on the percentiles of the target.
Left column: The differential operator $L_\pi$ considered in this work.
Right column: The differential operator $L_\pi^1$ considered in earlier work.]
}
\label{fig: Euclidean results 1}
\end{figure}

Results, in Fig. \ref{fig: Euclidean results 1} for $d=1$ and in Fig. \ref{fig: Euclidean results 2} in the Supplement for $d = 2$, showed that larger smoothness $\alpha$ leads to faster decay of worst case error.
In particular, we empirically observe convergence rates of $o(n^{-\frac{1}{2}})$.
In the case of the operator $L_\pi^1$ studied on the Euclidean manifold in \cite{Oates2018} (right column), it was proven (under some additional assumptions) that the worst case error decreases at $O(n^{-\frac{\alpha - 1}{d}})$ when $\pi$ is smooth, essentially because the kernel $k_{\pi,\sigma}$ has derivatives of order $\alpha - 1$.
This is consistent with the experimental results in Fig. \ref{fig: Euclidean results 1}.
A small extension of the theoretical methods used in \cite{Oates2018} gives a corresponding rate for the operator $L_\pi$ of $O(n^{-\frac{\alpha - 2}{d}})$, since $L_\pi$ is a second order differential operator.
The results in the left column of Fig. \ref{fig: Euclidean results 1} bear out this conjecture, with slower convergence of the worst case error for fixed $\alpha$ compared to the right column.

There was only a small difference between the worst case error in the first scenario ($\bm{x}_i \sim \mathcal{N}(0,1)$, top row) compared to the second scenario ($\bm{x}_i \sim \mathcal{N}(1,3)$, middle row).
This clearly illustrates the property that the points $\{\bm{x}_i\}_{i=1}^n$ need not form an approximation to $\mathcal{P}$ in the proposed method.
On the other hand, the stratified points (bottom row) appeared to mitigate the transient phase before the linear asymptotics kicks in, compared to the use of Monte Carlo points, and should be preferred.

In dimension $d=2$ the worst case error decays more slowly, consistent with the rates just conjectured.
Moreover, the asymptotic advantage of larger $\alpha$ is not clearly seen for $n \leq 10^2$ so that the transient phase appears to last for longer.
This is consistent with the well-known curse of dimension for isotropic kernel methods.

To investigate the robustness of the proposed method when the integrand $f$ is not well-approximated by elements in $H_{\pi,\sigma}$, we also report absolute integration errors in Fig. \ref{fig: Euclidean results 3} of the Supplement.
These additional experiments are brief, since they closely mirror the experiments conducted in \cite{Oates2018}.

\subsection{The Sphere $\mathbb{S}^2$} \label{subsec: sphere}

Next we considered arguably the most important non-Euclidean manifold; the sphere $\mathbb{S}^2$.

\paragraph{Differential Operator}

Recall that the coordinate map $\nu$ from Sec. \ref{sec: background} can be used to compute the metric tensor
$$
\mathrm{G} = \left( \begin{array}{cc} \sin^2 q_2 & 0 \\ 0 & 1 \end{array} \right)
$$
and a natural volume element $\mathrm{d}V = \sin q_2 \; \mathrm{d}q_1 \mathrm{d} q_2$.
It follows that, for a function $\phi : \mathbb{S}^2 \rightarrow \mathbb{R}$, we have the gradient differential operator
\begin{eqnarray*}
\nabla \phi & = & \frac{1}{\sin^2 q_2} \frac{\partial \phi}{\partial q_1} \partial_{q_1} + \frac{\partial \phi}{\partial q_2} \partial_{q_2} .
\end{eqnarray*}
Similarly, for a vector field $\underline{\bm{s}} = s_1 \partial_{q_1} + s_2 \partial_{q_2}$, we have the divergence operator
\begin{eqnarray*}
\nabla \cdot \underline{\bm{s}} & = & \frac{\partial s_1}{\partial q_1} + \frac{\partial s_2}{\partial q_2} + \frac{\cos q_2}{\sin q_2} s_2 .
\end{eqnarray*}
Thus the linear operator $L_\pi$ that we consider is:
\begin{eqnarray}
L_{\pi} (\phi) & = & \frac{\cos q_2}{\sin q_2} \frac{\partial \phi}{\partial q_2} + \frac{1}{\sin^2 q_2} \left\{ \frac{1}{\pi} \frac{\partial \pi}{\partial q_1} \frac{\partial \phi}{\partial q_1} + \frac{\partial^2 \phi}{\partial q_1^2} \right\}  + \left\{ \frac{1}{\pi} \frac{\partial \pi}{\partial q_2} \frac{\partial \phi}{\partial q_2} + \frac{\partial^2 \phi}{\partial q_2^2} \right\} . \label{eq: Lpi for S2}
\end{eqnarray}
Turning this into expressions in terms of $\bm{x}$ requires that we notice
$$
\frac{\cos q_2}{\sin q_2} = \frac{x_3}{\sqrt{1 - x_3^2}}, \qquad \frac{1}{\sin^2 q_2} = \frac{1}{1 - x_3^2}
$$
and use chain rule for partial differentiation (see Sec. \ref{ap: partial diff} in the Supplement).
This manifold-specific portion of MATLAB code is presented in Fig. \ref{fig: matlab 2} in the Supplement.

\paragraph{Choice of Kernel}

To proceed we require a reproducing kernel $k$ defined on $\mathbb{S}^2$.
To this end, Proposition 5 of \cite{Brauchart2013} establishes that, for $\alpha \in \mathbb{N} + \frac{1}{2}$, the kernel
\begin{eqnarray}
k(\bm{x},\bm{y}) \hspace{-5pt} & = & \hspace{-5pt} C^{(1)} {}_3 F_2\left[ \begin{array}{cc} \frac{3}{2}-\alpha, 1 - \alpha, \frac{3}{2} - \alpha \\ 2 - \alpha , 3 - \frac{m}{2} - 2\alpha \end{array} ; \frac{1 - \bm{x}\cdot\bm{y}}{2} \right] + C^{(2)}  \|\bm{x} - \bm{y}\|^{2\alpha-2}, \label{eq: S2 kernel}
\end{eqnarray}
defined for $\bm{x},\bm{y} \in \mathbb{S}^m$, reproduces the Sobolev space $H^\alpha(\mathbb{S}^m)$.
Here ${}_pF_q$ is  the generalised hypergeometric function.
The expressions for the constants $C^{(1)}$ and $C^{(2)}$ are given in Sec. \ref{ap: C1C2} of the Supplement.
This kernel was used in our experiments, with values $\alpha \in \{\frac{7}{2},\frac{9}{2},\frac{11}{2}\}$ considered.

\paragraph{Experimental Results}

\begin{figure}[t!]
\centering
\includegraphics[width = 0.8\textwidth,clip,trim = 5cm 4cm 5cm 1cm]{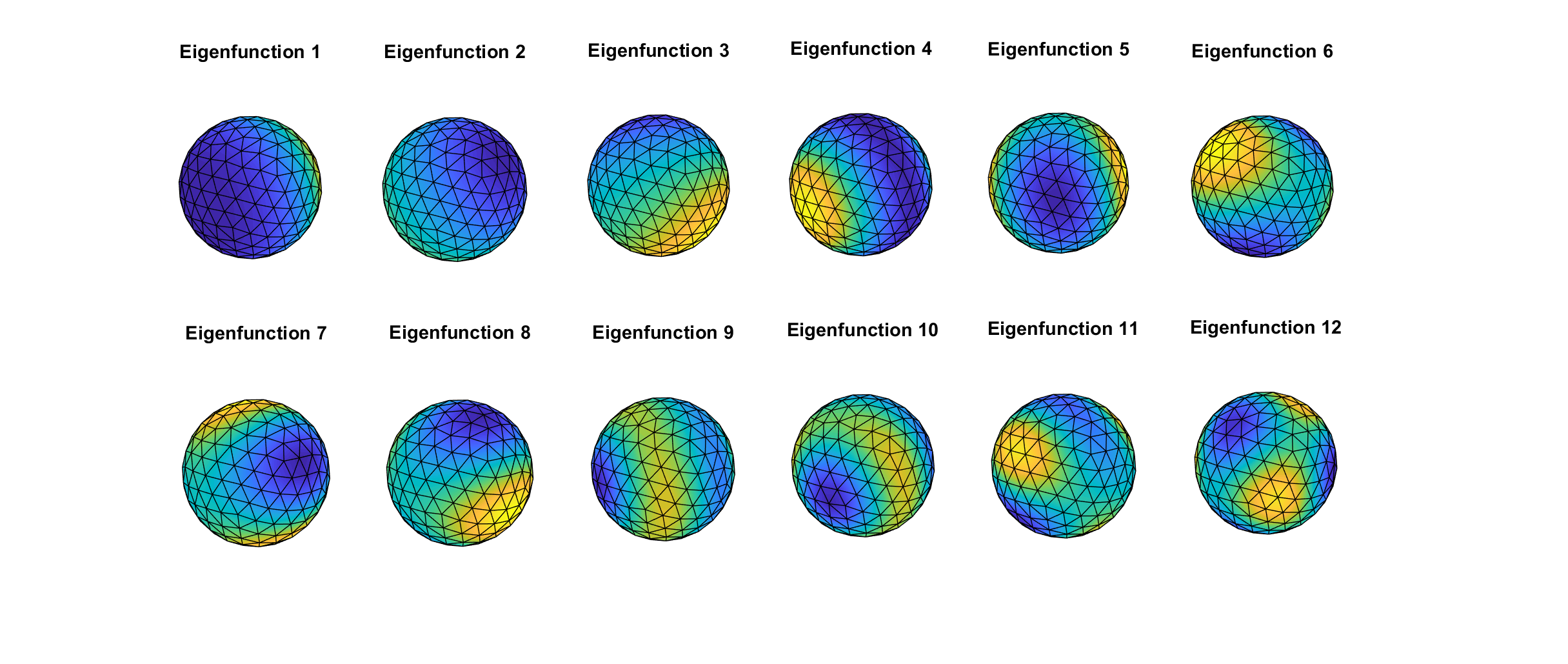}
\caption{The first 12 eigenfunctions of the kernel $k_\pi$ for a von Mises-Fisher distribution on $\mathbb{S}^2$.}
\label{fig: S2 eigenfunctions}
\end{figure}

To illustrate the method, considered the von Mises-Fisher distribution $\mathcal{P}$ whose density with respect to $\mathrm{d}V$ is
$$
\pi_{\mathcal{P}}(\bm{x}) = \frac{\|\bm{c}\|_2}{4 \pi \; \text{sinh}(\|\bm{c}\|_2)} \exp(\bm{c}^\top \bm{x}).
$$
For illustration we suppose that the normalisation constant is unknown and we have access only to $\pi(\bm{x})= \exp( \bm{c}^\top \bm{x} )$.
Thus we proceed to construct the differential operator operator $L_\pi$ as previously described.
To gain intuition as to the reasonableness of the resulting $k_\pi$, we fixed $\alpha = \frac{7}{2}$ and plotted the first few eigenfunctions of $k_\pi$ for the target distribution $\pi_{\mathcal{P}}$ defined by $\bm{c} = (1,0,0)^\top$.
These are shown in Fig. \ref{fig: S2 eigenfunctions}.
Note that, in simple visual terms, these seem like a reasonable basis in which to perform function approximation on $\mathbb{S}^2$.

Next, we considered the performance of the proposed integration method.
For various values of $n$, we obtained $n$ points $\{\bm{x}_i\}_{i=1}^n$ that were quasi-uniformly distributed on $\mathbb{S}^2$, being obtained by minimising a generalised electrostatic potential energy \citep[Reisz's-energy;][]{Semechko2015}.
Note that these points, being uniform, do not form an approximation to $\mathcal{P}$.
For each point set we computed the worst case error in Eqn. \ref{eq: wce}.

Results in Fig. \ref{fig: non-Euclid} (left) showed that the worst case error decays more rapidly with larger $\alpha$.
Again, we empirically observe convergence rates of $o(n^{-\frac{1}{2}})$.
Although it is not possible to accurately read off asymptotic rates from these results, they are somewhat consistent with a convergence rate of $O(n^{-\frac{\alpha - 2}{d}})$ where $d = 2$.

To investigate the robustness of the proposed method when the integrand $f$ is not itself an element of $H_{\pi,\sigma}$, we also report absolute integration errors in Sec. \ref{ap: S2 numerics extra} of the Supplement.

\begin{figure}[t!]
\centering
\includegraphics[width = 0.33\textwidth,clip,trim = 6cm 10.5cm 6cm 11cm]{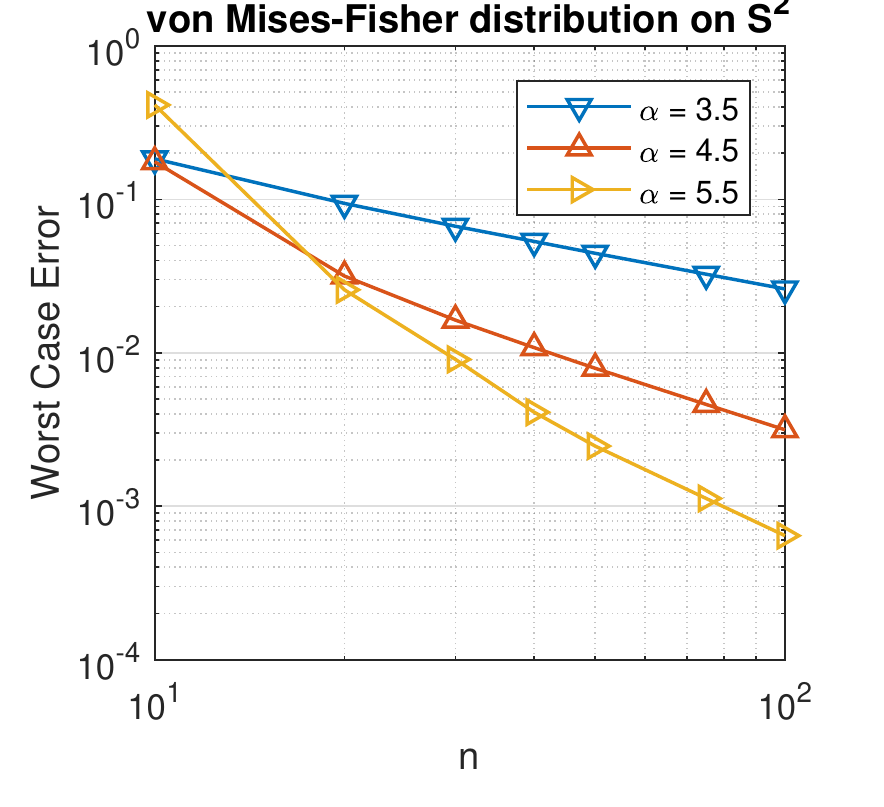} 
\includegraphics[width = 0.33\textwidth,clip,trim = 6cm 10.5cm 6cm 11cm]{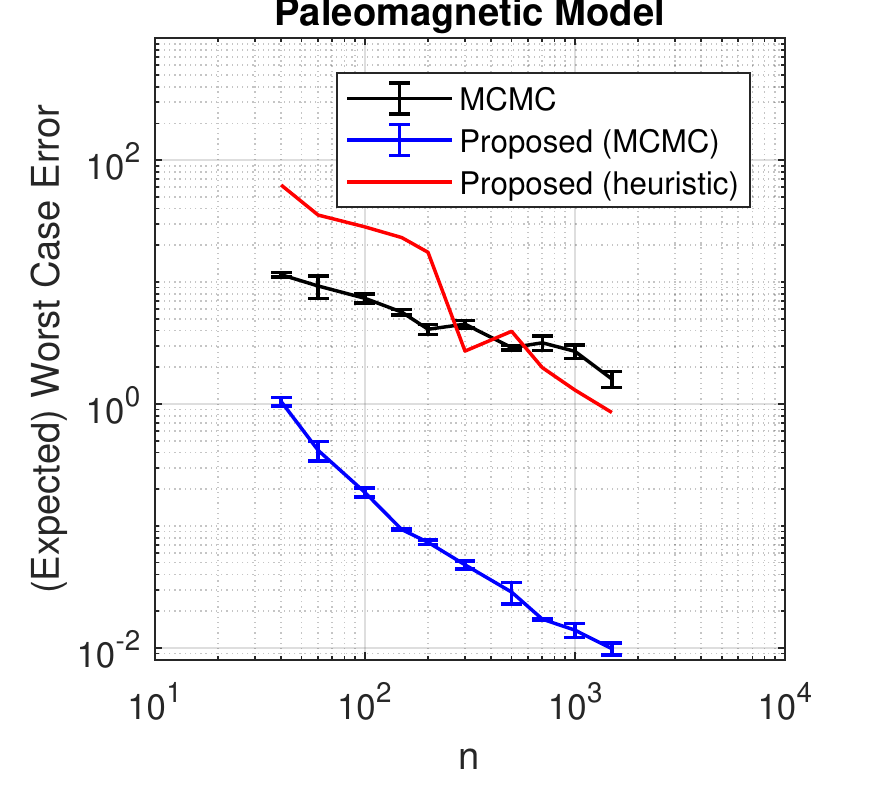} 

\caption{Results for non-Euclidean manifolds.
Left: The von Mises-Fisher distribution on $\mathbb{S}^2$.
[The worst case error of the proposed method (Eqn. \ref{eq: wce}) was plotted for various $\alpha$, controlling the smoothness of the kernel, and various $n$, the number of evaluations of the integrand.
The points $\bm{x}_i$ were quasi-uniform over $\mathbb{S}^2$ and the differential operator $L_\pi$ was used.]
Right: Results for the paleomagnetic model.
[The worst case error of the proposed method (Eqn. \ref{eq: wce}) was plotted for  various $n$, the number of evaluations of the integrand.
Black: The ergodic average from MCMC.
Blue: The proposed method based on the same MCMC point set.
Red: The proposed method based on a heuristic, whereby points were approximately stratified on the manifold.]}
\label{fig: non-Euclid}
\end{figure}

\subsection{Australian Mesozoic Paleomagnetic Pole Model} \label{subsec: Application}

The data that we considered consists of $n=33$ independent estimates $\{\bm{y}_i\}_{i=1}^n \subset \mathbb{S}^2$ for the position of the Earth's historic magnetic pole \citep[Table 2 of][]{Schmidt1976}.
The task is to aggregate these estimates; for details see \cite{Paine2017}.
The statistical model considered uses the von Mises-Fisher distribution as a likelihood:
\begin{eqnarray*}
\mathcal{L}(\bm{\mu} , \kappa) & \propto & \prod_{i=1}^n \frac{\kappa^{1/2}}{I_{\frac{1}{2}}(\kappa)} \exp(\kappa \bm{\mu}^\top \bm{y}_i)
\end{eqnarray*}
where $I_\eta$ is the modified Bessel function of the first kind of order $\eta$, $\bm{\mu} \in \mathbb{S}^2$ is a mean parameter and $\kappa \in \mathbb{R}_+$ is a concentration parameter.
Here the density is given with respect to the natural volume element $\mathrm{d}V$ on $\mathbb{S}^2$.
A Bayesian analysis proceeds under a conjugate prior
\begin{eqnarray*}
\pi_0(\bm{\mu} , \kappa) & \propto & \left(\frac{\kappa^{1/2}}{I_{\frac{1}{2}}(\kappa)}\right)^c \exp(R_0 \kappa \bm{\mu}_0^\top \bm{\mu})
\end{eqnarray*}
for hyper-parameters $\bm{\mu}_0 \in \mathbb{S}^2$, $c, R_0 \in \mathbb{R}_+$, given with respect to the natural volume element $\mathrm{d}V$ on the manifold $M = \mathbb{S}^2 \times \mathbb{R}_+$.
The task is to compute expectations under the posterior
\begin{eqnarray*}
\pi_{\mathcal{P}}(\bm{\mu},\kappa | \{\bm{y}_i\}_{i=1}^n) & \propto & \left( \frac{\kappa^{1/2}}{I_{\frac{1}{2}}(\kappa)} \right)^{c+n} \exp(R_n \kappa \bm{\mu}_n^\top \bm{\mu}),
\end{eqnarray*}
where 
\begin{eqnarray*}
R_n = \left\|R_0 \bm{\mu}_0 + \sum_{i=1}^n \bm{y}_i \right\|, \; \bm{\mu}_n = R_n^{-1} \left( R_0 \bm{\mu}_0 + \sum_{i=1}^n \bm{y}_i \right).
\end{eqnarray*}
In particular, we attempted to estimate the first and second moments of $\bm{\mu}$ and $\kappa$, so that $f(\bm{\mu},\kappa) = \mu_i^j$ for $i \in \{1,2,3\}$, or $f(\bm{\mu},\kappa) = \kappa^j$, in each case for $j \in \{1,2\}$.
The most default prior with $c = 0$, $R_0 = 0$ was employed, following \cite{Nunez-Antonio2005}.
Owing to the smoothness of these test functions, a version of the exponentiated-quadratic kernel was employed; see the Supplement for full detail.

Points $\bm{x}_i = (\bm{\mu}_i,\kappa_i)$ were sampled from the posterior via random-walk MCMC.
The usual ergodic average was used to estimate each integral and the result was compared to the estimate provided by the proposed method based on the same point set.
Results in Figs. \ref{fig: non-Euclid} (right) showed that the proposed method strongly outperformed the ergodic average for integrands in the unit ball of $H_\pi$.
In Fig. \ref{fig: pole res 2} the proposed method was seen to performed well for the test functions $\mu_i^j$, but delivered similar performance to the ergodic average for the test functions $\kappa^j$.
This may be because the latter were not well-approximated by elements of $H_\pi$.

For interest, we also considered a heuristic approach where points were approximately uniformly stratified, as described in Sec. \ref{ap: paleo supp sec} of the Supplement.
Note that these points do not form an empirical approximation to the posterior; however, for the proposed method this is not required.
Results in Fig. \ref{fig: non-Euclid} (right) showed similar performance to MCMC.

\begin{figure}[t!]
\centering
\includegraphics[width = 0.8\textwidth,clip,trim = 1.7cm 9cm 1.8cm 9cm]{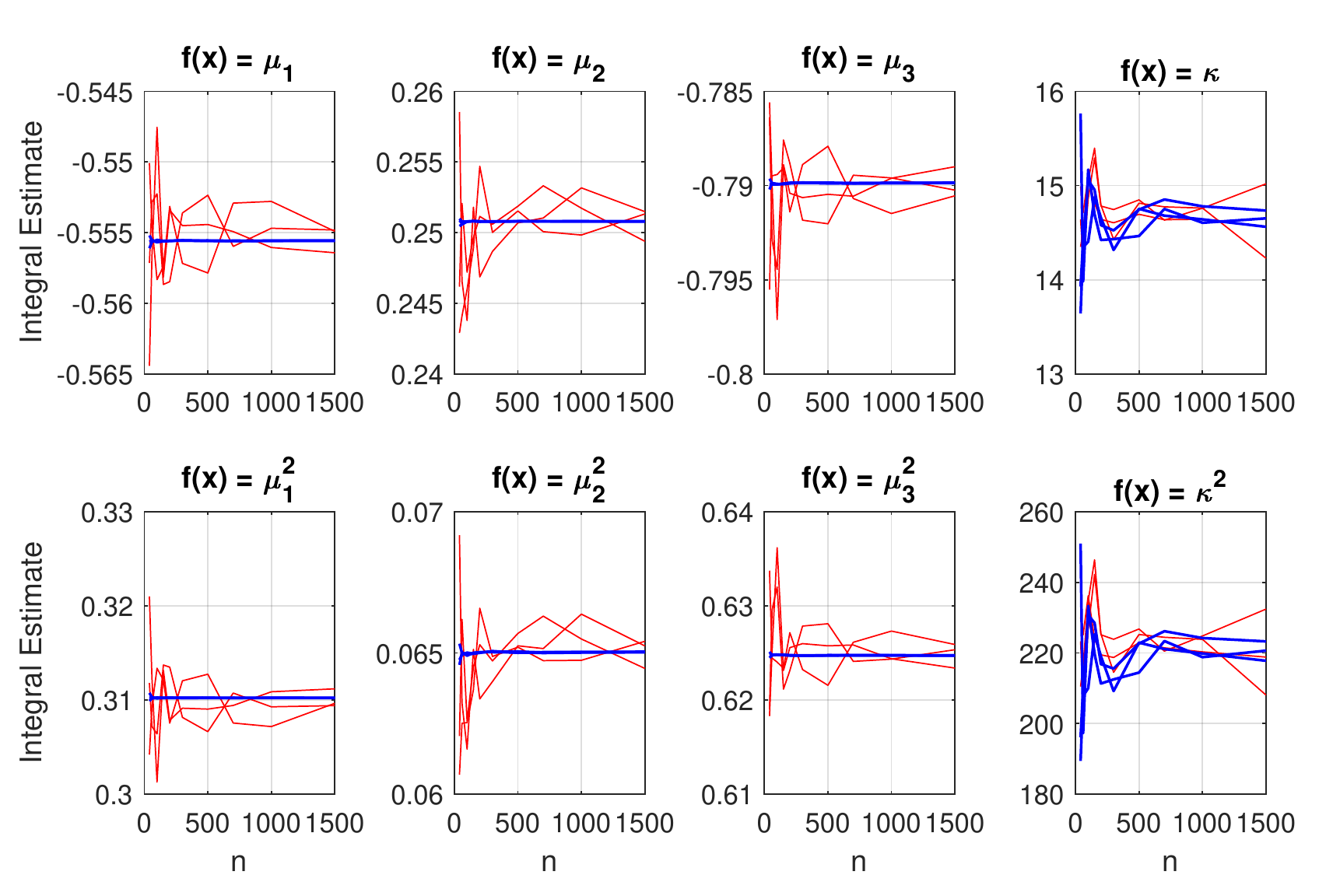}
\caption{Results for the paleomagnetic model.
Integral estimates were obtained and plotted for various $n$, the number of evaluations of the integrand.
[Red: The ergodic average from MCMC.
Blue: The proposed method based on the same MCMC point set.
Results from 3 independent chains are shown; in most panels the proposed method has indistinguishable variation, so the blue curves cannot be distinguished.]
}
\label{fig: pole res 2}
\end{figure}

\section{Discussion} \label{sec: discussion}

This paper generalised of the method of \cite{Oates2017} to a general oriented Riemannian manifold.
The method was illustrated for regular integrals of modest dimension; as usual, the case of high-dimensional manifolds (i.e. $m$ large) is likely to challenge any regression-based method unless strong assumptions can be made on the integrand.
Three open theoretical questions are: (1) How expressive is the function class $H_{\pi,\sigma}$ in terms of approximation properties? (2) How quickly do these estimators converge under particular regularity assumptions on the integrand? (3) How robust are these estimators when the function space assumptions are violated?
The ultimate success of this method will hinge on the extent to which these important questions can be addressed.

\section*{Acknowledgements}

CJO and MG were supported by the Lloyd's Register Foundation programme on data-centric engineering at the Alan Turing Institute, UK.
AB was supported by a Roth scholarship from the Department of Mathematics at Imperial College London, UK.
MG was supported by the EPSRC grants [EP/K034154/1, EP/R018413/1, EP/P020720/1, EP/L014165/1], and an EPSRC Established Career Fellowship, [EP/J016934/1]. 
CJO is grateful for discussions with Chang Liu.

\newpage
\appendix

\section{Supplement}

This electronic supplement contains several theoretical details and empirical results that were referenced in the paper \emph{Posterior Integration on a Riemannian Manifold}.

\subsection{Proofs} \label{ap: proof sec}

This section contains proofs of Theorems \ref{thm: main} and \ref{thm: limit} from the main text.

\begin{proof}[Proof of Theorem \ref{thm: main}]
By definition, $\pi = \mathcal{L} \pi_0$ and so we have that $Z \mathrm{d}\mathcal{P} = \pi \mathrm{d}V$.
Thus
\begin{eqnarray*}
\int_M L_\pi(\phi) \mathrm{d} \mathcal{P} & = & \int_M L_\pi(\phi) \frac{\pi}{Z} \mathrm{d}V \\
& = & \int_M \frac{\nabla \cdot (\pi \nabla \phi)}{\pi} \frac{\pi}{Z} \mathrm{d}V \\
& = & \frac{1}{Z} \int_M \nabla \cdot (\pi \nabla \phi) \mathrm{d}V 
\end{eqnarray*}
where $\mathrm{d}V$ is the natural volume form on $M$.
From the divergence theorem we have that:
\begin{eqnarray*}
\int_M \nabla \cdot (\pi \nabla \phi) \mathrm{d}V & = & \int_{\partial M} \langle \pi \nabla \phi , \underline{\bm{n}} \rangle_{\mathrm{G}} \; i_{\underline{\bm{n}}} \mathrm{d}V 
\end{eqnarray*}
which establishes the result that was claimed.
\end{proof}

\begin{proof} [Proof of Theorem \ref{thm: limit}]
Note that $\mathbf{K}_{\pi,\sigma} = \sigma^2 \mathbf{1} \mathbf{1}^\top + \mathbf{K}_\pi$.
The proof is then an application of the Woodbury matrix inversion formula, which can be used to deduce that
\begin{eqnarray*}
\lim_{\sigma \rightarrow \infty} \int_M \hat{f} \mathrm{d}\mathcal{P} & = & \lim_{\sigma \rightarrow \infty} \sigma^2 \mathbf{1}^\top (\sigma^2 \mathbf{1} \mathbf{1}^\top + \mathbf{K}_\pi)^{-1} \mathbf{f} \\
& = & \lim_{\sigma \rightarrow \infty} \frac{\mathbf{1}^\top \mathbf{K}_\pi^{-1} \mathbf{f}}{\sigma^{-2} + \mathbf{1}^\top \mathbf{K}_\pi^{-1} \mathbf{1}} \\
& = & \left( \frac{\mathbf{K}_\pi^{-1} \mathbf{1}}{\mathbf{1}^\top \mathbf{K}_\pi^{-1} \mathbf{1}} \right)^\top \mathbf{f}
\end{eqnarray*}
as required.
\end{proof}

\subsection{Alternative Differential Operators} \label{ap: stein operators}

Alternatives to the differential operator $L_\pi$ can also be considered.
To this end, recall vector fields $\underline{\bm{s}} = s_1 \partial_{q_1} + \dots + s_m \partial_{q_m}$ are differential operators, so that we can consider the directional derivative of a function $f : M \rightarrow \mathbb{R}$ in the direction $\underline{\bm{s}}$, denoted $\underline{\bm{s}}(f) = s_1 \partial_{q_1} f + \dots + s_m \partial_{q_m} f$. 
Now, note that $\nabla \cdot (f \underline{\bm{s}}) =  \underline{\bm{s}}(f) + f \nabla \cdot \underline{\bm{s}}$. 
In particular, if $\underline{\bm{s}} = \nabla h$, then $\nabla \cdot (f \nabla h) = f \Delta h + (\nabla h) (f) $. 

From the above identities we have that, for a closed manifold $M$, vector field $\underline{\bm{s}}$ and function $f$,
$$
\int_M \underline{\bm{s}}(f) +f \nabla  \cdot \underline{\bm{s}} \; \mathrm{d}V =0.
$$
The operator $L_{\pi}$ in the main text is thus the special case with $\underline{\bm{s}} = \pi \nabla \phi$ and $f=1$. Another possibility is $f=\pi$ and $\underline{\bm{s}}=\nabla \phi$:
$$
\int_M (\nabla \phi)(\pi) +\pi \Delta \phi \; \mathrm{d}V = 0.
$$
Similarly, we also have Green's identity
$$
\int_M  f \Delta g - g \Delta f \; \mathrm{d}V = 0  
$$
In particular, if $f$ satisfies the Poisson equation $\Delta f = \rho$, then 
$$ 
\int_M  f \Delta g - g\rho \; \mathrm{d}V = 0
$$
and if $f$ is harmonic ($\Delta f=0$), then 
\begin{eqnarray}
\int_M f \Delta g \; \mathrm{d}V = 0 . \label{eq: ab alternative}
\end{eqnarray}
This suggests other possible differential operators, for example if we take $g := \pi \phi$ in Eqn. \ref{eq: ab alternative} we obtain a differential operator $\tilde{L}_{\pi}(\phi) = \frac{f \Delta(\pi \phi)}{\pi}$ for any harmonic $f$. Of course, any linear combination of the above operators integrates to $0$ as well.

In this work we do not exhibit a criteria under which one operator may be preferred to another, although we note that the related discussion in \cite{Gorham2016} may be relevant.
However, a computational preference should clearly be afforded to differential operators that are of lower order; for this reason we would prefer $L_\pi$ in the main text to $\tilde{L}_\pi$ above, as the former requires only first order derivatives of $\pi$.

\subsection{Additional Results for the Euclidean Manifold}

The experiment presented in the main text was repeated in dimension $d = 2$ and the results are presented in Fig. \ref{fig: Euclidean results 2}.

\begin{figure}[h!]
\centering
\includegraphics[width = 0.33\textwidth,clip,trim = 6cm 10.5cm 6cm 11cm]{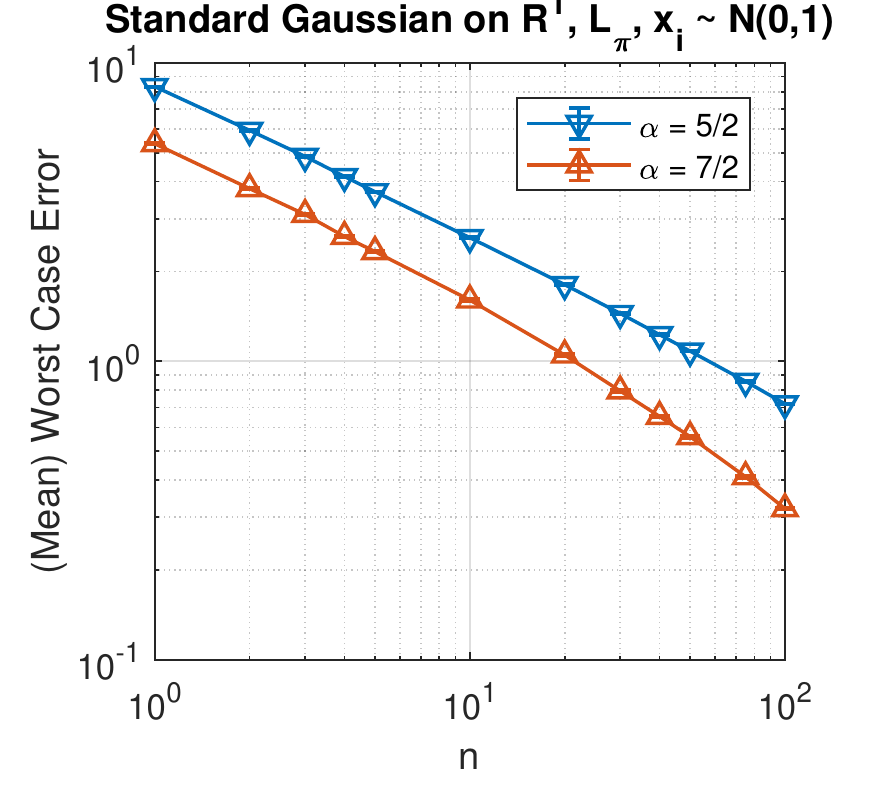}
\includegraphics[width = 0.33\textwidth,clip,trim = 6cm 10.5cm 6cm 11cm]{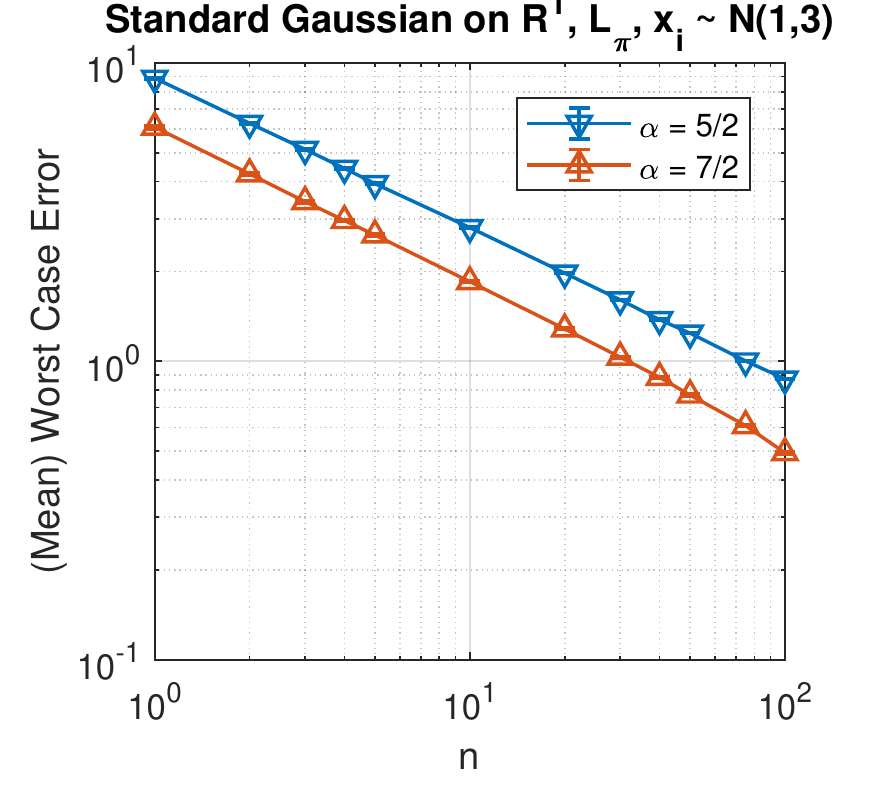}

\caption{Results for the standard Gaussian on the Euclidean manifold in dimension $d = 2$.
The worst case error of the proposed method (Eqn. \ref{eq: wce}) was plotted for various $\alpha$, controlling the smoothness of the kernel, and various $n$, the number of evaluations of the integrand.
[Left: The points $\bm{x}_i \sim \mathcal{N}(0,1)$ were drawn from the target.
Right: The points $\bm{x}_i \sim \mathcal{N}(1,3)$ were drawn from an incorrect distribution.
In each case the differential operator $L_\pi$ was used.]
}
\label{fig: Euclidean results 2}
\end{figure}

To assess the robustness of the proposed method when the integrand $f$ need not lie close to the set $H_{\pi,\sigma}$, we considered integration of explicit test functions $f(\bm{x})$.
In Fig. \ref{fig: Euclidean results 3} we fixed $d = 1$, $a = \frac{7}{2}$ and generated $n$ independent samples from $\mathcal{Q} = \mathcal{N}(1,3)$.
The proposed method based on $L_\pi$ was employed to integrate $f(x) = x^j$ for $j \in \{1,2\}$ and the integral estimates so-obtained were compared to those provided by the importance sampling Monte Carlo estimator:
$$
\frac{1}{n} \sum_{i=1}^n f(\bm{x}_i) \frac{\pi_{\mathcal{P}}(\bm{x}_i)}{\pi_{\mathcal{Q}}(\bm{x}_i)}
$$
Of course, the true integrals in this case are, respectively, $0$ and $1$.
It is seen that the proposed method provided lower variance estimation than the Monte Carlo method.

\begin{figure}[h!]
\centering
\includegraphics[width = 0.55\textwidth,clip,trim = 3cm 10cm 3cm 9cm]{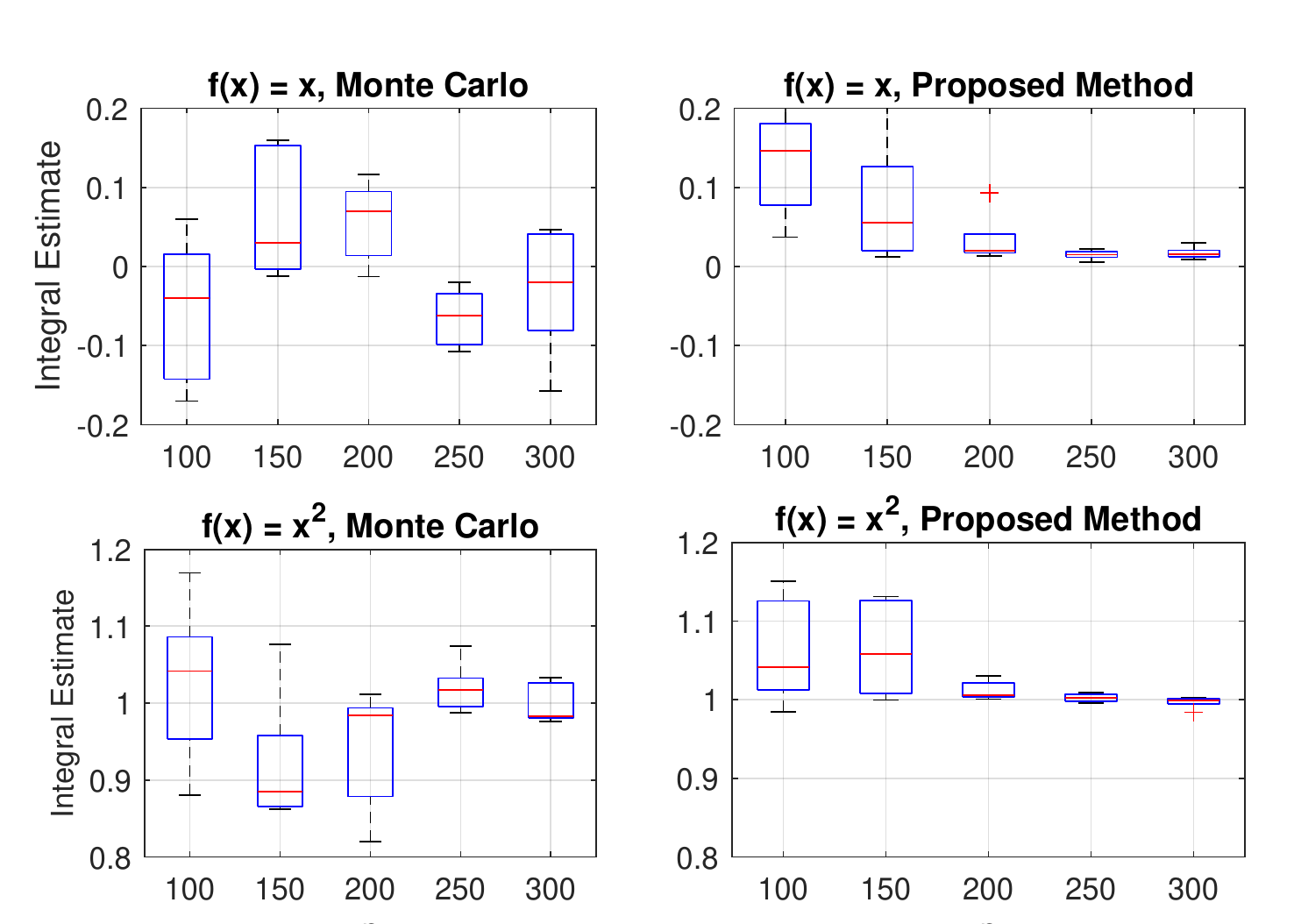}

\caption{Additional results for $M = \mathbb{R}$.
Integral estimates were obtained and plotted for various $n$, the number of evaluations of the integrand.
The true integral for $f(x) = x$ is 0 and for $f(x) = x^2$ is 1.
Several realisations of the point set were obtained and the associated estimates were aggregated into a boxplot.
[Black: The importance sampling Monte Carlo method.
Blue: The proposed method based on the same point set.]
}
\label{fig: Euclidean results 3}
\end{figure}

\subsection{Chain Rule for Partial Differentiation} \label{ap: partial diff}

The following identities are useful, converting differential operators in local coordinates $\bm{q}$ into differential operators in global coordinates $\bm{x}$:
\begin{eqnarray}
\frac{\partial}{\partial q_i} & = & \sum_{j=1}^m \frac{\partial x_j}{\partial q_i} \frac{\partial}{\partial x_j} \label{eq: chain rule 1} \\
\frac{\partial^2}{\partial q_i^2} & = & \sum_{j,k=1}^m \frac{\partial x_j}{\partial q_i} \frac{\partial x_k}{\partial q_i} \frac{\partial^2}{\partial x_j \partial x_k} + \sum_{j=1}^3 \frac{\partial^2 x_j}{\partial q_i^2} \frac{\partial}{\partial x_j}  \label{eq: chain rule 2}
\end{eqnarray}
For example, an application of these identities to Eqn. \ref{eq: Lpi for S2} leads to the differential operator used in the code snippet in Fig. \ref{fig: matlab 2}.

\subsection{Choice of Kernel on $\mathbb{S}^2$} \label{ap: C1C2}

The constant terms in the kernel in Eqn. \ref{eq: S2 kernel} are as follows:
\begin{eqnarray*}
C^{(1)} & = & \frac{2^{2\alpha - 2}}{2\alpha - 2} \frac{(\frac{m}{2})_{2\alpha - 2}}{(m)_{2\alpha - 2}} \\
C^{(2)} & = & (-1)^{\alpha - \frac{1}{2}} 2^{1 - 2\alpha} \frac{\Gamma(\frac{m+1}{2}) \Gamma(\alpha - \frac{1}{2}) \Gamma(\alpha - \frac{1}{2})}{\sqrt{\pi} \Gamma(\frac{m}{2}) (\frac{1}{2})_{\alpha - \frac{1}{2}} (\frac{m}{2})_{\alpha - \frac{1}{2}} } .
\end{eqnarray*}
Here $(z)_n := \Gamma(z+n) / \Gamma(z)$ is the Pochhammer symbol.
See \cite{Brauchart2013} for full detail.

\subsection{Additional Results for $\mathbb{S}^2$} \label{ap: S2 numerics extra}

To assess the robustness of the proposed method when the integrand $f$ need not lie close to the set $H_{\pi,\sigma}$, we considered integration of explicit test functions $f(\bm{x})$.
In Fig. \ref{fig: sphere results 2} we fixed $\alpha = \frac{7}{2}$ and generated a quasi-uniform point set of size $n$, as described in the main text.
The proposed method based on $L_\pi$ was employed to integrate $f(x) = x^j$ for $j \in \{1,2\}$ and the integral estimates so-obtained were compared to those provided by the importance sampling Monte Carlo estimator:
$$
\frac{\sum_{i=1}^n f(\bm{x}_i) \pi(\bm{x}_i)}{\sum_{i=1}^n \pi(\bm{x}_i)}
$$
From symmetry, the true integrals of $x_2$ and $x_3$ are, respectively, $0$ and $1$.

Results are presented in Fig. \ref{fig: sphere results 2}.
It is seen that the proposed method provided lower variance estimation than the Monte Carlo method.

\begin{figure}[h!]
\centering
\includegraphics[width = 0.55\textwidth,clip,trim = 3cm 7.5cm 3cm 6.5cm]{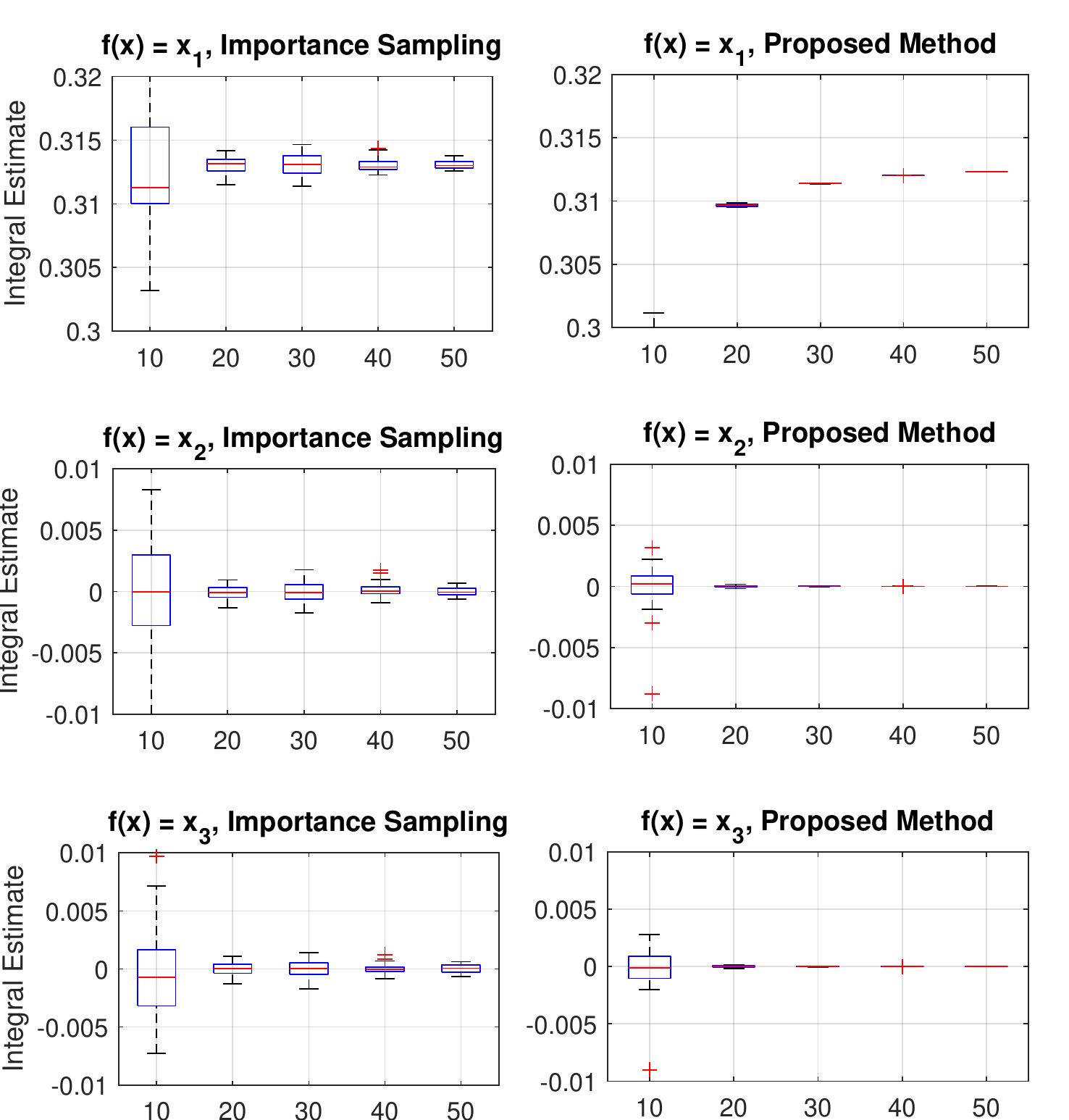}

\includegraphics[width = 0.55\textwidth,clip,trim = 2.8cm 7cm 3.2cm 6.4cm]{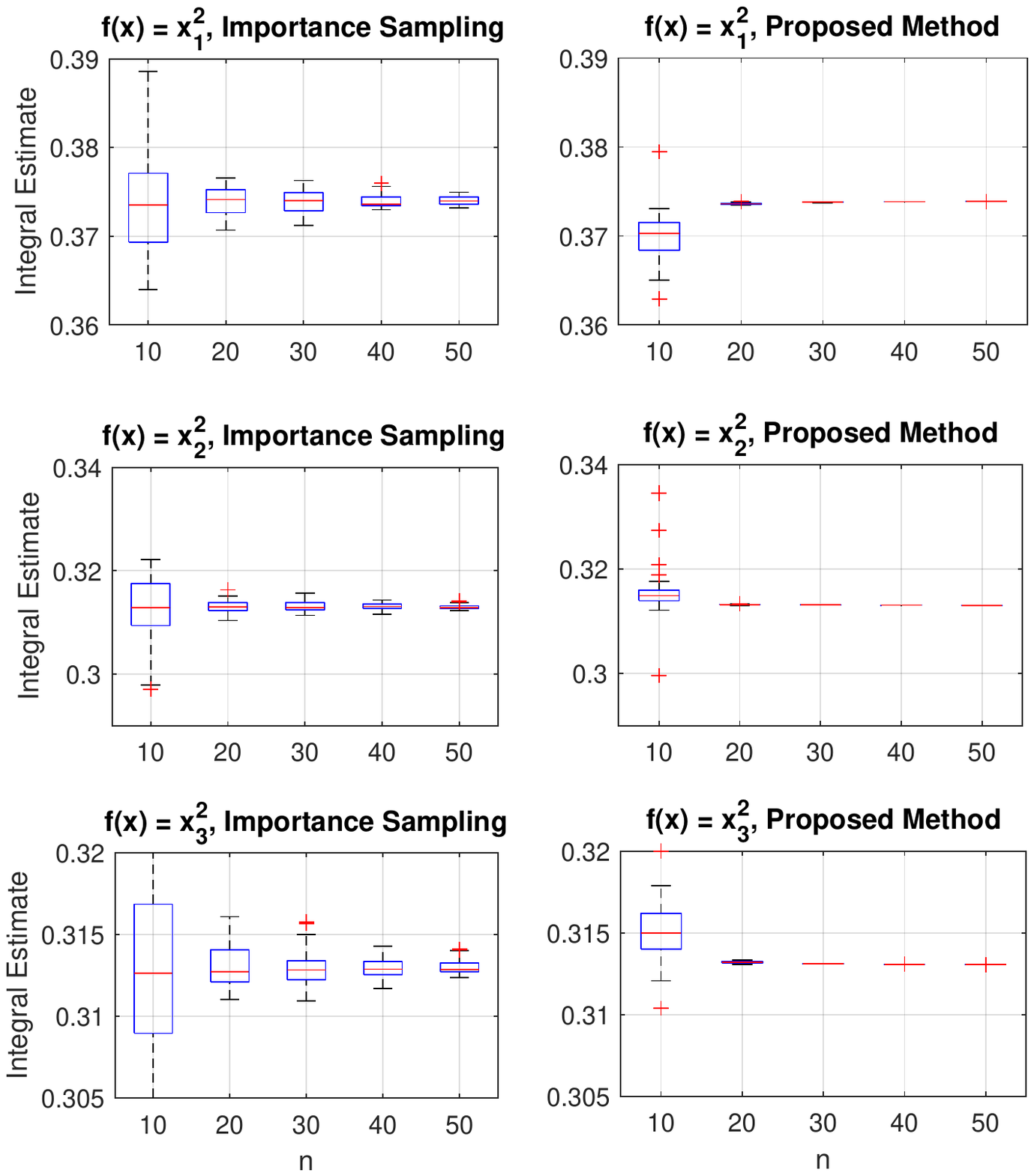}

\caption{Additional results for $M = \mathbb{S}^2$.
Integral estimates were obtained and plotted for various $n$, the number of evaluations of the integrand.
Several realisations of the point set were obtained and the associated estimates were aggregated into a boxplot.
[Black: The importance sampling Monte Carlo method.
Blue: The proposed method based on the same point set.]
}
\label{fig: sphere results 2}
\end{figure}

\subsection{Derivations for Paleomagnetic Model} \label{ap: paleo supp sec}

Consider the local coordinate system $\bm{q} = (q_1,q_2,q_3)$ such that
$$
\bm{x} = \nu(\bm{q}) = (\cos q_1 \sin q_2 , \sin q_1 \sin q_2 , \cos q_2 , q_3)
$$
where $q_1 \in [0,2\pi)$, $q_2 \in [0,\pi]$ and $q_3 \in \mathbb{R}_+$.
Thus $M$ can be considered a $m=3$ dimensional manifold embedded in $\mathbb{R}^4$.

\paragraph{Differential Operator}

From the coordinate map $\nu$ we compute the metric tensor
$$
\mathrm{G} = \left( \begin{array}{ccc} \sin^2 q_2 & 0 & 0 \\ 0 & 1 & 0 \\ 0 & 0 & 1 \end{array} \right) .
$$
This leads to a volume element $\sin q_2 \; \mathrm{d}q_1 \mathrm{d} q_2 \mathrm{d}q_3$.
It follows that, for a function $\phi : M \rightarrow \mathbb{R}$, we have the gradient differential operator
\begin{eqnarray*}
\nabla \phi & = & \frac{1}{\sin^2 q_2} \frac{\partial \phi}{\partial q_1} \partial_{q_1} + \frac{\partial \phi}{\partial q_2} \partial_{q_2} + \frac{\partial \phi}{\partial q_3} \partial_{q_3} .
\end{eqnarray*}
Similarly, for a vector field $\underline{\bm{s}} = s_1 \partial_{q_1} + s_2 \partial_{q_2} + s_3 \partial_{q_3}$, we have the divergence operator
\begin{eqnarray*}
\nabla \cdot \underline{\bm{s}} & = & \frac{\partial s_1}{\partial q_1} + \frac{\partial s_2}{\partial q_2} + \frac{\partial s_3}{\partial q_3} + \frac{\cos q_2}{\sin q_2} s_2 .
\end{eqnarray*}
Thus the linear operator $L_\pi$ that we consider is:
\begin{eqnarray*}
L_{\pi} (\phi) & = & \frac{\cos q_2}{\sin q_2} \frac{\partial \phi}{\partial q_2} + \frac{1}{\sin^2 q_2} \left\{ \frac{1}{\pi} \frac{\partial \pi}{\partial q_1} \frac{\partial \phi}{\partial q_1} + \frac{\partial^2 \phi}{\partial q_1^2} \right\} \\
& & + \left\{ \frac{1}{\pi} \frac{\partial \pi}{\partial q_2} \frac{\partial \phi}{\partial q_2} + \frac{\partial^2 \phi}{\partial q_2^2} \right\} + \left\{ \frac{1}{\pi} \frac{\partial \pi}{\partial q_3} \frac{\partial \phi}{\partial q_3} + \frac{\partial^2 \phi}{\partial q_3^2} \right\} .
\end{eqnarray*}
Turning this into expressions in terms of $\bm{x}$ requires that we notice
$$
\frac{\cos q_2}{\sin q_2} = \frac{x_3}{\sqrt{1 - x_3^2}}, \qquad \frac{1}{\sin^2 q_2} = \frac{1}{1 - x_3^2}
$$
and use chain rule for partial differentiation as in Eqns. \ref{eq: chain rule 1} and \ref{eq: chain rule 2}.
This manifold-specific portion of MATLAB code is presented in Fig. \ref{fig: matlab 3}.

\paragraph{Choice of Kernel}

To proceed, we require a reproducing kernel $k$ defined on $M = \mathbb{S}^2 \times \mathbb{R}_+$.
However, $M$ has been embedded in the ambient space $\mathbb{R}^4$ and we can induce a kernel on $M$ as the restriction of a kernel on $\mathbb{R}^4$.
Since we know that our test functions $f$ are each infinitely differentiable, we elected to use a smooth (exponentiated quadratic) kernel:
\begin{eqnarray*}
k\left( \left[ \begin{array}{c} \bm{\mu} \\ \kappa \end{array} \right] , \left[ \begin{array}{c} \bm{\mu}' \\ \kappa' \end{array} \right] \right) & = & \tan^{-2}(\kappa) \tan^{-2}(\kappa') \exp\left(- \frac{1}{2} \left\| \left[ \begin{array}{c}\bm{\mu} - \bm{\mu}' \\ \kappa - \kappa' \end{array} \right] \right\|^2 \right)
\end{eqnarray*}
The first terms are needed to ensure vanishing of the integral over $\partial M$ in the divergence theorem.
In this case, the condition is that $\partial \phi / \partial \kappa$ vanishes at $\kappa = 0$.
Note that the choice of arctan is not unique, and any smooth function $t(\kappa)$ could be used provided that $t'(\kappa) = 0$ whenever $\kappa = 0$. 

\paragraph{Experimental Results}

In addition to points $\bm{x}_i = (\bm{\mu}_i,\kappa_i)$ obtained via MCMC, we also considered a stratified point set.
To be specific, we considered a tensor-product design where $O(n^{2/3})$ basis points were quasi-uniformly distributed over $\mathbb{S}^2$ as in Sec \ref{subsec: sphere} and $O(n^{1/3})$ basis points were stratified according to the quantiles of the density 
\begin{eqnarray}
\tilde{\pi}(\kappa) \propto \left(\frac{\kappa^{1/2}}{I_{1/2}(\kappa)} \right)^{c+n} \exp( R_n \kappa) . \label{eq: ptilde}
\end{eqnarray}
The tensor product of these two bases provided a point set of size $O(n)$.
The density in Eqn. \ref{eq: ptilde} is a heuristic; it is not exactly related to the posterior $\kappa$ marginal, since it holds $\bm{\mu} = \bm{\mu}_n$ fixed, but is perhaps expected to be a reasonable approximation to this marginal.

\newpage
\onecolumn
\subsection{Symbolic Differentiation in MATLAB} \label{ap: sym diff}

In this section we provide code snippets that demonstrate the use of symbolic differentiation in computation of the kernel $k_\pi$.
Fig. \ref{fig: matlab} is from the Euclidean case $M = \mathbb{R}^d$ with $d = 2$, whilst Fig. \ref{fig: matlab 2} is for the sphere $M = \mathbb{S}^2$ and Fig. \ref{fig: matlab 3} is for the paleomagnetic model with $M = \mathbb{S}^2 \times \mathbb{R}_+$.

For Fig. \ref{fig: matlab}, it should be noted that only lines 5-12 depend on the geometry of the manifold, and these are independent of both $\pi$ and $k$.
Thus, generic code for (e.g.) the sphere $\mathbb{S}^2$ and other manifolds can be provided.
Figs. \ref{fig: matlab 2} and \ref{fig: matlab 3} therefore include just the code that is specific to differentiation on their manifold.

\begin{figure*}[h!]
\begin{lstlisting}
% logarithm of the un-normalised measure pi
log_pi = @(x1,x2) - x1^2 - x2^2;

% Differential operator L_pi on the Euclidean manifold R^2
dq1 = @(f,x1,x2) diff(f,x1); % d/dq_1
dq2 = @(f,x1,x2) diff(f,x2); % d/dq_2
d2q1 = @(f,x1,x2) diff(f,x1,2); % d^2/dq_1^1
d2q2 = @(f,x1,x2) diff(f,x2,2); % d^2/dq_2^2
L = @(f,x1,x2) dq1(log_pi(x1,x2),x1,x2)*dq1(f,x1,x2) ...
               + dq2(log_pi(x1,x2),x1,x2)*dq2(f,x1,x2) ...
               + d2q1(f,x1,x2) ...
               + d2q2(f,x1,x2);

% Radial basis function (alpha = 5/2)
matern = @(r) (1 + sqrt(5)*r + 5*r^2/3) * exp(-sqrt(5)*r);

% Reproducing kernel
syms x1 x2 y1 y2
k = matern(sqrt((x1-y1)^2 + (x2-y2)^2));

% Symbolic differentiation
L_k = L(k,x1,x2); % differentiate wrt [x1,x2]
L_Lbar_k = L(L_k,y1,y2); % differentiate wrt [y1,y2]
\end{lstlisting}
\caption{Symbolic differentiation was used to automate computation of the kernel $k_\pi$. [This MATLAB R2017b code snippet is for the problem considered in Section \ref{subsec: Euclidean} for $d=2$ dimensions. The differential operator was $L_\pi$.]}
\label{fig: matlab}
\end{figure*}

\begin{figure*}[h!]
\begin{lstlisting}
% Differential operator L_pi on the sphere \mathbb{S}^2
dq1 = @(f,x1,x2,x3) -x2*diff(f,x1) + x1*diff(f,x2); % d/dq_1
dq2 = @(f,x1,x2,x3) x1*x3*(1-x3^2)^(-1/2)*diff(f,x1) ...
                    + x2*x3*(1-x3^2)^(-1/2)*diff(f,x2) ...
                    - (1-x3^2)^(1/2)*diff(f,x3); % d/dq_2
d2q1 = @(f,x1,x2,x3) x2^2*diff(f,x1,2) ...
                     - 2*x1*x2*diff(f,x1,x2) ...
                     + x1^2*diff(f,x2,2) ...
                     - x1*diff(f,x1) ...
                     - x2*diff(f,x2); % d^2/dq_1^2
d2q2 = @(f,x1,x2,x3) x1^2*x3^2*(1-x3^2)^(-1)*diff(f,x1,2) ...
                     + 2*x1*x2*x3^2*(1-x3^2)^(-1)*diff(f,x1,x2) ...
                     - 2*x1*x3*diff(f,x1,x3) ...
                     + x2^2*x3^2*(1-x3^2)^(-1)*diff(f,x2,2) ...
                     - 2*x2*x3*diff(f,x2,x3) ...
                     + (1-x3^2)*diff(f,x3,2) ...
                     - x1*diff(f,x1) ...
                     - x2*diff(f,x2) ...
                     - x3*diff(f,x3); % d^2/dq_2^2
L = @(f,x1,x2,x3) x3*(1-x3^2)^(-1/2)*dq2(f,x1,x2,x3) ...
        + (1-x3^2)^(-1)*dq1(log_pi(x1,x2,x3),x1,x2,x3)*dq1(f,x1,x2,x3) ...
                  + (1-x3^2)^(-1)*d2q1(f,x1,x2,x3) ...
                  + dq2(log_pi(x1,x2,x3),x1,x2,x3)*dq2(f,x1,x2,x3) ...
                  + d2q2(f,x1,x2,x3); 
\end{lstlisting}
\caption{Symbolic differentiation was used to automate computation of the kernel $k_\pi$. [This MATLAB R2017b code snippet is for the problem considered in Section \ref{subsec: sphere}, the sphere $\mathbb{S}^2$. The differential operator was $L_\pi$.]}
\label{fig: matlab 2}
\end{figure*}

\begin{figure*}[h!]
\begin{lstlisting}
% Differential operator L_pi on the manifold \mathbb{S}^2 x R_+
dq1 = @(f,x1,x2,x3,x4) -x2*diff(f,x1) + x1*diff(f,x2); % d/dq_1
dq2 = @(f,x1,x2,x3,x4) x1*x3*(1-x3^2)^(-1/2)*diff(f,x1) ...
                       + x2*x3*(1-x3^2)^(-1/2)*diff(f,x2) ...
                       - (1-x3^2)^(1/2)*diff(f,x3); % d/dq_2
dq3 = @(f,x1,x2,x3,x4) diff(f,x4); % d/dq_3
d2q1 = @(f,x1,x2,x3,x4) x2^2*diff(f,x1,2) ...
                        - 2*x1*x2*diff(f,x1,x2) ...
                        + x1^2*diff(f,x2,2) ...
                        - x1*diff(f,x1) ...
                        - x2*diff(f,x2); % d^2/dq_1^2
d2q2 = @(f,x1,x2,x3,x4) x1^2*x3^2*(1-x3^2)^(-1)*diff(f,x1,2) ...
                        + 2*x1*x2*x3^2*(1-x3^2)^(-1)*diff(f,x1,x2) ...
                        - 2*x1*x3*diff(f,x1,x3) ...
                        + x2^2*x3^2*(1-x3^2)^(-1)*diff(f,x2,2) ...
                        - 2*x2*x3*diff(f,x2,x3) ...
                        + (1-x3^2)*diff(f,x3,2) ...
                        - x1*diff(f,x1) ...
                        - x2*diff(f,x2) ...
                        - x3*diff(f,x3); % d^2/dq_2^2
d2q3 = @(f,x1,x2,x3,x4) diff(f,x4,2); % d^2/dq_3^2
L = @(f,x1,x2,x3,x4) x3*(1-x3^2)^(-1/2)*dq2(f,x1,x2,x3,x4) ...
  + (1-x3^2)^(-1)*dq1(logp(x1,x2,x3,x4),x1,x2,x3,x4)*dq1(f,x1,x2,x3,x4) ...
               + (1-x3^2)^(-1)*d2q1(f,x1,x2,x3,x4) ...
               + dq2(logp(x1,x2,x3,x4),x1,x2,x3,x4)*dq2(f,x1,x2,x3,x4) ...
               + d2q2(f,x1,x2,x3,x4) ...
               + dq3(logp(x1,x2,x3,x4),x1,x2,x3,x4)*dq3(f,x1,x2,x3,x4) ...
               + d2q3(f,x1,x2,x3,x4); 

\end{lstlisting}
\caption{Symbolic differentiation was used to automate computation of the kernel $k_\pi$. [This MATLAB R2017b code snippet is for the problem considered in Section \ref{subsec: Application}, the manifold $\mathbb{S}^2 \times \mathbb{R}_+$. The differential operator was $L_\pi$.]}
\label{fig: matlab 3}
\end{figure*}

\end{document}